\documentclass[11pt, a4paper]{article}
\usepackage[margin=25mm]{geometry}

\usepackage{authblk}

\usepackage{amsthm,amsmath,amsfonts,amssymb}
\usepackage{graphicx}
\usepackage{enumitem}
\usepackage[authoryear]{natbib}

\usepackage[colorlinks,citecolor=blue,urlcolor=blue]{hyperref}

\usepackage{subcaption}
\usepackage{array}
\usepackage{bm}
\usepackage{bbm}
\usepackage[ruled,linesnumbered]{algorithm2e}
\usepackage{color}
\usepackage{multirow}
\usepackage{rotating}

\usepackage{booktabs}
\usepackage{setspace}
\usepackage{siunitx}

\usepackage{anyfontsize}

\usepackage{todonotes}

\theoremstyle{plain}
\newtheorem{theorem}{Theorem}
\newtheorem{lemma}[theorem]{Lemma}
\newtheorem{proposition}[theorem]{Proposition}
\newtheorem{corollary}[theorem]{Corollary}

\theoremstyle{remark}
\newtheorem{definition}{Definition}
\newtheorem{remark}{Remark}

\newtheorem{condition}{Condition}
\newtheorem{example}{Example}

\providecommand{\keywords}[1]
{
  \small	
  \textbf{\textit{Keywords---}} #1
}

\newcommand{\dd}{\mathtt{d}}
\newcommand{\mbb}[1]{\mathbb{#1}}
\newcommand{\mcal}[1]{\mathcal{#1}}
\newcommand{\mtt}[1]{\mathtt{#1}}

\newcommand{\bmx}{\bm{x}}
\newcommand{\bmy}{\bm{y}}
\newcommand{\bmX}{\bm{X}}
\newcommand{\bmXi}{\bm{\Xi}}
\newcommand{\bmY}{\bm{Y}}
\newcommand{\bmD}{\bm{\nabla}}
\newcommand{\bmW}{\bm{W}}

\newcommand{\indi}{\mathbb{I}}

\newcommand{\bmc}{\bm{c}}

\newcommand{\blue}{\color{blue}}
\newcommand{\splitcell}[2]{\begin{tabular}{@{}#1@{}} #2 \end{tabular}}

\newcommand{\vmf}{\text{vMF}}

\newcommand{\Vol}{\text{Vol}}

\newcommand{\changed}{\blue}

\DeclareMathOperator*{\argmin}{argmin\,}

\DeclareMathOperator*{\GenLog}{GenLog}
\DeclareMathOperator*{\Var}{Var}
\DeclareMathOperator*{\Skew}{Skew}
\DeclareMathOperator*{\kurt}{XsKurt}

\title{Statistical Disaggregation --- a Monte Carlo Approach for Imputation under Constraints}

\author[1]{Shenggang~Hu}
\author[2]{Hongsheng~Dai}
\author[3]{Fanlin~Meng}
\author[4]{\authorcr Louis~Aslett}
\author[2]{Murray~Pollock}
\author[1]{Gareth~O.~Roberts}

\affil[1]{Department of Statistics,
University of Warwick, Coventry, CV4 7AL, UK}
\affil[2]{School of Mathematics, Statistics and Physics, Newcastle University, NE1 7RU, UK}
\affil[3]{University of Exeter Business School, University of Exeter, Exeter, EX4 4PU, UK}
\affil[4]{Department of Mathematical Sciences, Durham University, DH1 3LE, UK}

\date{}

\begin{document}
\maketitle
\begin{abstract}
Equality-constrained models naturally arise in problems in which measurements are taken at different levels of resolution. The challenge in this setting is that the models usually induce a joint distribution which is intractable. Resorting to instead sampling from the joint distribution by means of a Monte Carlo approach is also challenging. For example, a naive rejection sampling does not work when the probability mass of the constraint is zero.
A typical example of such constrained problems is to learn energy consumption for a higher resolution level based on data at a lower resolution, e.g., to decompose a daily reading into readings at a finer level.
We introduce a novel Monte Carlo sampling algorithm based on Langevin diffusions and rejection sampling to solve the problem of sampling from equality-constrained models. 
Our method has the advantage of being exact for linear constraints and naturally deals with multimodal distributions on arbitrary constraints.
We test our method on statistical disaggregation problems for electricity consumption datasets, and our approach provides better uncertainty estimation and accuracy in data imputation compared with other naive/unconstrained methods.




\end{abstract}

\section{Introduction}


\subsection{Motivation}
Consider the following sampling problem from a constrained joint distribution
\begin{equation}
f_{\mcal{H}}\left(\bmy^{(1)},\dots,\bmy^{(m)}\right)\propto f_1(\bmy^{(1)})f_2(\bmy^{(2)})\cdots f_m(\bmy^{(m)}) \mbb{I}_{\bm{y}^{(1:m)}\in\mcal{H}}\label{eq:cons1}
\end{equation}
where, for simplicity, $\bmy^{(i)}\in\mbb{R}^d$ are vectors of the same dimension and $f_i(y^{(i)}):\mbb{R}^d\rightarrow\mbb{R}_{>0}$ are strictly positive continuous density functions on $\mbb{R}^d$.
The joint distribution $f_{\mcal{H}}$ in (\ref{eq:cons1}) is constrained on the level set $\mcal{H}:=\{\bmy^{(1:m)}\in\mbb{R}^{md}:h(\bmy^{(1:m)})=0\}$ defined by some function $h:\mbb{R}^{md}\rightarrow\mbb{R}^{k}$.
Under mild regularity conditions of $h$, the level set will form a Riemannian manifold embedded in the Euclidean space $\mbb{R}^{md}$.
When the Riemannian metric $g$ associated with the Riemannian manifold $\mcal{H}$ is fixed, then $(\mcal{H},g)$ admits a unique canonical measure which we use as the dominating measure of the density in (\ref{eq:cons1}).
When we fix the Riemannian metric associated with $\mcal{H}$ to be the Riemmanian metric induced by the Euclidean space of $y^{(1)},\dots,y^{(m)}$, 
 (Please refer to Appendix \ref{appx:manifold} for more details.)
Note that an ad-hoc rejection sampling does not work on $f_\mcal{H}$ since $f_{\mcal{H}}$ integrates to zero with respect to the Lebesgue measure on $\mbb{R}^{md}$.

In general, sampling from (\ref{eq:cons1}) is not trivial even if the constraint $\mcal{H}$ is linear, since the constrained distribution is usually not tractable apart from a handful of special cases such as a Gaussian distribution constrained on a hyperplane or a hypersphere.
However, constrained problems arise naturally in Statistics for both linear and non-linear constraints.
In fact, one can find various settings where the simple linear, sum or average, constraint inherently resides in the model, whenever measurements are taken at different levels of resolutions.
For instance, in energy consumption time series \citep{Peppanen.2016} where the consumption is recorded in different time resolution, in spatial statistics \citep{li2023hierarchical} where the average of measurements in the fine-grained map should match with the measurements in the coarse-grained map, in survey sampling calibrations \citep{deville1992calibration} where the calibration weights $w_k$ are computed using auxiliary variables $x_k$ such that the sample statistics $\sum_k w_k x_k$ matches the population statistics $T$, {\changed etc..}
In essence, the goal is to model the unknown values of finer resolution conditioned on knowing exactly the corresponding aggregated value of lower resolution.

The prediction of high-resolution data from low-resolution data is often named as \textbf{\emph{disaggregation}} in some literature \citep{wang2020regional,rafsanjani2020load}, and when the unknown values are missing data, then this is often termed as \textbf{\emph{data imputation}}.
Throughout this paper, we refer to such problems captured by (\ref{eq:cons1}) as {\slshape disaggregation} or {\slshape imputation} in general, but when talking about the statistical models without constraint we still refer to them as {\slshape forecasting models} or {\slshape predictors}, and we refer to the high-resolution data estimated from the low-resolution data as {\slshape imputed data}.

More concretely, in the settings of energy consumption time series for instance, an energy supplier usually has customers with different types of meters installed \citep{meng2018integrated}, e.g., customers with smart meters (record energy consumption from every hour to every minute), customers with time-of-use meters (e.g., Economy 7 in the UK that a day can be divided into two time periods and the meter records aggregated consumption over the two periods), and customers with traditional meters (record aggregated consumption). 
It poses a great challenge for the energy suppliers to fully understand energy customers' consumption patterns, especially for the latter two types of customers with conventional meters in the absence of high time-resolution meter data. 
Even though it is supposed to be easy to know the detailed and high-resolution energy consumption of customers with smart meters, the smart meter data may still be subject to delays and lower reliability \citep{Peppanen.2016}, or may be aggregated to preserve customers' privacy.
For example in the UK, the smart meter data that distributed network operators receive will be an aggregated reading without the real-time data \citep{Poursharif.2017}.
Therefore, the supplier often has low-resolution data for some (usually recent) periods, but possibly high-resolution data for the periods before. 
For those days with missing high-resolution data from a customer, nonetheless, the energy supplier may still want to know the more fine-grained energy consumption of such customers.
For instance, knowing consumption during peak time periods for customers with traditional meters or consumption during each hour for customers with time-of-use meters helps understand energy usage behaviours, which is essential to transform the energy systems in industrialized countries in order to reduce the total energy consumption \citep{customerbehavior.2015}. 


A similar problem can be found in power distribution networks where a grid operator would like to understand when spikes of energy demand could occur.
This study would require a continuous recording of energy usage in the network at a fine-grained level. 
Such detailed monitoring needs the installation of additional equipment and storage devices which can be expensive.
However, such expenses can be avoided if one can reasonably predict the peak and trough measurements in each time period given the low-frequency data.

We may summarize this problem under the following framework.\footnote{Code deposited \href{https://github.com/Shenggang/linearly_constrained_fusion}{here}.} 
Let $\boldsymbol Y_t = (Y^{(1)}_t,\cdots, Y^{(m)}_t)$ denote the high-resolution data for time period $t$, where the high-resolution data is the result of naturally dividing each low-resolution data into $m$ readings.
Our target is to impute the missing high-resolution data $\boldsymbol Y_t$ for time period $t$ from the existing data set ${\cal D}_t = \{\boldsymbol Y_k, k=1,\cdots, t-1 \}$ and a set of additional covariates $\bmXi_t$ containing information related to $\bm{Y}_t$, under some equality constraint, in many cases linear, for instance, $\sum_{i=1}^{m} Y^{(i)}_t = S_t$. 
Here $S_t$ corresponds to the low-resolution aggregated reading for time period $t$ and it is available when we impute the high-resolution readings. 
Therefore, if we denote the imputed values as $\hat{\boldsymbol Y}_t = (\hat Y^{(1)}_t,\cdots, \hat Y^{(m)}_t)$, it must satisfy the constraint $\sum_{i=1}^{m} \hat Y^{(i)}_t = S_t$ too. 
In one of the application problems of this paper, we consider $m=3$, i.e., peak time period (evening), off-peak time period (midnight), and day-time, since these are of most interests to electricity providers and different tariffs are often made on these time periods.
The data vector $\boldsymbol Y_t = (Y^{(1)}_t,\cdots, Y^{(m)}_t)$ follows a density $\prod_i f(y^{(i)}_t|\boldsymbol \theta, {\cal D}_t, \bmXi)$.
For $t$ fixed, $Y^{(i)}_t, i=1,\cdots, m$ are assumed to be independent conditioned on all historical data and covariates. 


In many cases, the observation data are not well-captured by Gaussian models, thus sampling from (\ref{eq:cons1}) is often not trivial in reality even if the constraint is linear.
For electricity consumption data, the residual distribution based on time series models usually will not be Gaussian because of the extreme values, e.g., due to abnormal weather conditions.
Considering the dataset used in Section \ref{sec:study1}, the Irish Smart Meter Trial data \citep{CER1,CER2}, we fitted an auto-regressive time series model using the 2009 autumn season data to avoid seasonal components and the fitted error is presented in Fig. \ref{fig:error_qq} as a quantile-quantile plot, where the error points are plotted against its normal estimation.
It is clear from the graph that the Gaussian assumption for residuals is not appropriate (also evidenced by the Shapiro-Wilk test having a p-value $<2.2 \times 10^{-16}$). 

\begin{figure}[t]
    \centering
    \includegraphics[width=0.65\linewidth]{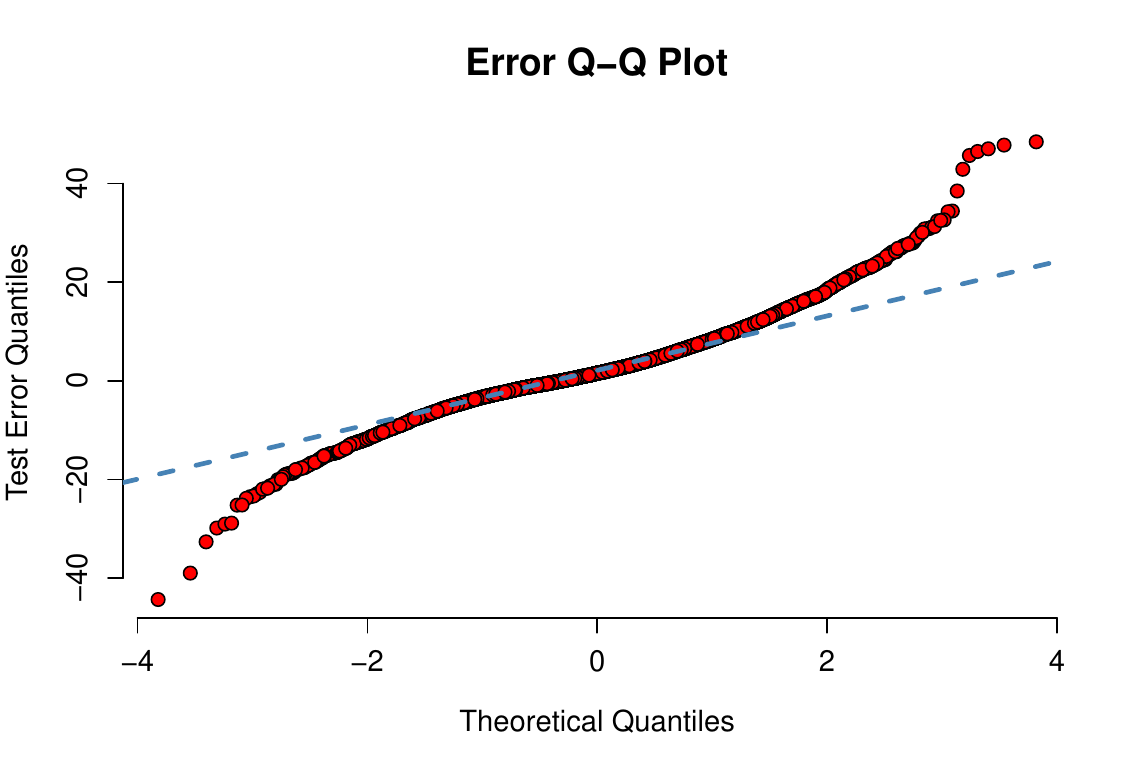}
    \caption{Residual distribution plot based on AR model (without constraints) for the Irish Smart Meter Trial dataset from Study 1 (Sec. \ref{sec:study1}). An AR(7) model is fitted on the 2009 autumn season consumption data to avoid seasonal components. The residual plot indicates non-Gaussian error.}
    \label{fig:error_qq}
\end{figure}

\subsection{Why and When to Consider Disaggregation}
To understand how adding the constraint can benefit the unconstrained model, we begin with a simple multivariate Gaussian model and study the difference in the model uncertainty and mean-squared error (MSE). 
The result can be informally posed as:
\begin{enumerate}
    \item The total uncertainty of the {\bf constrained} model is guaranteed to be {\bf less} than the total uncertainty of the {\bf unconstrained} model.
    \item When the original model has a large uncertainty compared with its bias, the model MSE will be improved when incorporating the constraint.
\end{enumerate}
Moreover, a simulation study shows the second result remains effective on some other unimodal distributions when the condition in the proposition is met.
The detailed analysis and proof of the above two points are presented in Appendix \ref{appx:MSE}. 

In the next section, we will review some of the methods that one may use to conduct inference on a constrained problem and point out their drawbacks.
We present our Constrained Fusion algorithm in Section \ref{sec:MCF}, and compare our method with three other more conventional sampling approaches in Section \ref{sec:comparison}.
The results reveal our method has a faster convergence rate in difficult situations when the target is heavy-tailed and lies away from the constraint.
We then formulate a model for the imputation problem we described above and show how our algorithm can be applied.
We also illustrate how our methodology works on disaggregating constrained time series on real datasets in Section \ref{sec:application_framework} with two more examples in Appendix \ref{appx:application}.
Our data analysis result shows combining constraints could indeed improve the accuracy of imputation by mainly reducing its uncertainty, agreeing with what we see theoretically in Appendix \ref{appx:MSE}.
The paper ends with a discussion in Section \ref{sec:discussion}.

\section{Background}
Returning to the general case, recall from (\ref{eq:cons1}) that we want to sample from a product density subject to a certain equality constraint
\begin{equation*}
f_{\mcal{H}}\left(\bmy^{(1)},\dots,\bmy^{(m)}\right)\propto f_1(\bmy^{(1)})f_2(\bmy^{(2)})\cdots f_m(\bmy^{(m)}) \mbb{I}_{\bm{y}^{(1:m)}\in\mcal{H}}
\end{equation*}
In fact, it is fair to question whether the above density is well-defined, which we will discuss in Appendix \ref{appx:manifold} the (sufficient) conditions when $f_\mcal{H}$ is properly a density function with respect to a dominating measure on the set $\mcal{H}$.


Suppose for now that (\ref{eq:cons1}) is well-defined with respect to a base measure in $\mcal{H}$.
The difficulty in implementing an MCMC algorithm on a constrained sampling problem lies in generating proposals that satisfy the constraint, usually by means of projection or transformation onto the constraint set.
\cite{zappa2018monte} and \cite{chua2020sampling} both consider generating proposals by first sampling from the tangential plane and then projecting onto the manifold.
\cite{chua2020sampling} presents a way to efficiently expand the base sample set to $n\times m$ weighted samples that are approximately distributed as the target distribution.
\cite{zappa2018monte} presents a modified Metropolis-Hastings (MH) algorithm where the proposals are generated through tangential projections onto the manifold. 
This algorithm resembles the usual random walk MH algorithm in that each proposal has a relatively low cost to generate and the rejection rate is directly related to the step size. 
Since the proposal generation depends on the result of an iterative solver, there is an additional rejection stage for reverse projection check to ensure every step is reversible, i.e., the iterative solver can also move from the proposal back to the current state.

The other branch builds upon the Hamiltonian Monte Carlo (HMC) method, where proposals are generated from a simulated Hamiltonian system. 
CHMC \citep{brubaker2012family} extends HMC where samples are generated by including the constraint into the Hamiltonian system and the evolution of which is solved by a constrained integrator.
The problem with CHMC is that the usual explicit integrator cannot be adopted as there is a constraint on the system, and the implicit integrator requires an iterative solver for each simulation step. 
A special case is Geodesic HMC \citep{byrne2013geodesic} which splits the Hamiltonian system such that the integrator avoids the need for an iterative solver for simulating the Hamiltonian mechanics, given that the geodesic flow can be exactly computed. 
This approach can be applied in directional statistics where the state space is usually an $n$-sphere for which the geodesics are explicitly known. 
Another problem of HMC is that the simulated Hamiltonian system needs to be reversible up to momentum reversal for detailed balance to hold.
Although such reversibility is usually satisfied for sufficiently small time steps, this might be violated when the parameter is tuned for more efficient simulation.

Recently, \cite{dai2017new} developed a rejection sampling approach based on Langevin diffusion bridges, whereby a diffusion is simulated subject to a constraint on the ending point. Further, \cite{dai2019monte} developed the Monte Carlo Fusion (MCF) algorithm, which simulates multiple diffusion bridges that coalesce into a single ending point, such that the marginal distribution of the endpoints is distributed exactly as the target distribution.
The MCF algorithm can be viewed as sampling from (\ref{eq:cons1}) subject to the constraint that all components take the same value.
In this paper, we extend these ideas to handle arbitrary constraints in (\ref{eq:cons1}). 
The new method employs $m$ Langevin diffusions, which start from a value at time $0$ following the distribution $f_i$, and their ending points at time $T$ follow a Gaussian distribution with the required constraints.
Therefore, the original non-Gaussian constraint problem becomes a Gaussian constraint problem. 
Finally, the outcomes based on Gaussian constraints will be adjusted according to a path-space rejection sampling for the Langevin diffusion processes. 
This adjusted ending point at $T$ exactly satisfies the constraint and follows the required target distribution. 
Based on the simulated Monte Carlo samples, we can obtain estimated statistics of interest, e.g., mean, variance, quantile points, etc.
The ability to simulate the samples exactly, through algorithms like MCF \cite{dai2019monte}, is of significant importance in practice since the samples are i.i.d. and there is no need to assess convergence like in Markov Chain Monte Carlo methods.

Perfect sampling (even for approximate sampling) under equality constraints is genuinely hard even for simple constraints like linear ones, with a couple of exceptions such as Gaussian distribution under linear constraints \citep{cong2017fast, vrins2018sampling}, which is not {suitable to} apply to {skewed} data such as discussed herein. 
\cite{allard2015disaggregating} addressed a similar problem of disaggregation with respect to linear constraint.
However, their approach is to approximate the constraint by allowing Monte Carlo samples close enough to the constraint to be accepted. As a consequence, the acceptance rate diminishes quickly with the error margin. 
In contrast, our approach innately ensures the samples always land on the constraint.


\section{Methodology}\label{sec:MCF}
It is hard to simulate directly from (\ref{eq:cons1}) 
due to the support having a lower dimension than the unconstrained state space.
However, when simulating the diffusion process, we can restrict the endpoints at time $T$ (typically conditional Gaussian under the proposal distribution) to satisfy the constraint and consider the probability law of the diffusion bridge conditioned on the endpoints instead.
This way, it is easier to construct the sampling method, since the target distribution $\prod f_i$ and the constraint $\mcal{H}$ are essentially satisfied separately at two independent stages.
In this section, we will first discuss the target and proposal distributions before presenting the full algorithm.



\subsection{Constructing Target and Proposal Diffusions} \label{sec:construction}
To begin with, we consider the following augmented distribution which leads us to the constrained product density in (\ref{eq:cons1}).
\textbf{Informally}, we state the following proposition.
\begin{proposition}[Informal]\label{proposition:proposition_conditional}
Consider a set of $m$ diffusion processes of length $T$, with the transition kernel $p_i(\bmX_T|\bmX_0),\,i=1,2,\cdots,m$ such that process $i$ admits $f_i^2(\cdot)$ as its invariant distribution. 
Then the joint density defined on the space $\mbb{R}^{md}\times\mcal{H}$
\begin{eqnarray}
g_{\mcal{H}}\left(\bm{x}^{(1)},\cdots, \bm{x}^{(m)}, \bm{y}^{(1)},\cdots,\bm{y}^{(m)}\right) \propto 
\prod_{i=1}^m f_i^2(\bm{x}^{(i)}) p_i(\bm{y}^{(i)}|\bm{x}^{(i)}) \frac{1}{f_i(\bm{y}^{(i)})}\indi_{\bmy^{(1:m)}\in \mcal{H}}. 
\label{eq:biased_langevin}
\end{eqnarray}
admits the constrained target density (\ref{eq:cons1}) as the marginal distribution of the ending points $\left(\bmy^{(1)}, \cdots, \bmy^{(m)}\right)$.
\end{proposition}
Since the initial points $\bmx^{(i)}$ follow the invariant distribution of the processes, the conclusion follows directly from integrating out all the $\bmx^{(i)}$s in (\ref{eq:biased_langevin}).

\vspace{1em}\noindent
To avoid the technical details, we will assume that there exists a canonical choice for the dominating measure on $\mcal{H}$ and integrating (\ref{eq:biased_langevin}) on $\mcal{H}$ is a well-defined operation
and address the measure on manifold in more detail in Appendix \ref{appx:manifold}.
To help with understanding, consider a one-dimensional diffusion connecting $X_0=x$ and $X_T=y$. Define the constraint $\mcal{H}$ as the single point set $\mcal{H}=\{y: y=y^*\}$ for a given value $y^*$. 
Then computing $\int_{\mbb{R}}f^2(x) p(y^*|x)\tfrac{1}{f(y^*)}\dd x$ gives the normalizing constant for (\ref{eq:biased_langevin}), which means
(\ref{eq:biased_langevin}) is well-defined as long as $f(y^*)\neq 0$.
Notice that imposing/altering the constraint on $\bmy^{(1:m)}$ only affects the normalizing constant of (\ref{eq:biased_langevin}) but not the transition kernel $p_i$ nor the initial distribution $f_i^2$, due to the way we decompose the diffusion measure.
\begin{remark}
In Proposition \ref{proposition:proposition_conditional}, the reason that $f_i^2(\cdot)$ is chosen as the invariant distribution (instead of $f_i(\cdot)$) is to cancel out the extra $(f_i(x^{(i)}))^{-1}$ term introduced by the Girsanov formula when computing the transition kernel $p_i(\bmy^{(i)}|\bmx^{(i)})$.
\end{remark}
To construct such processes with transition kernel $p_i(\boldsymbol y|\boldsymbol x)$ in Proposition \ref{proposition:proposition_conditional}, we consider the following. 
Let $\bmX^{(i)}:=\{\bmX_s^{(i)}:s \in [0,T]\}$ be a $d$-dimensional Langevin diffusion process  with transition kernel $p_i(\bmX_T|\bmX_0),\,i=1,2,\cdots,m$, defined as 
\begin{equation}
 \dd \bm{X}_s^{(i)}=\bmD \log f_i(\bmX_s^{(i)})\dd s + \dd \bmW_s^{(i)}, \label{eq:langevin_process}
\end{equation}
where $T$ is a constant, $\bmW^{(i)}$ is a $d$-dimensional Brownian motion, $\bmD$ is the gradient operator.
By \cite{hansen2003geometric}, $\bm{X}^{(i)}$ has invariant distribution proportional to $f_i^2(\bm{x})$ over $[0,T]$.
Such diffusion processes can be simulated by using Brownian bridges as the proposal diffusion.
More importantly, we can simulate the proposals exactly with the constraint applied to the ending points.
Define the proposal distribution $h_{\mcal{H}}:\mbb{R}^{md}\times \mcal{H}\rightarrow \mbb{R}_{>0}$ as
\begin{equation}
h_{\mcal{H}}(\bm{x}^{(1)},\cdots\bm{x}^{(m)},\bm{y}^{(1)},\cdots,\bmy^{(m)}) \propto \prod_{i=1}^m f_i(\bmx^{(i)})(2\pi T)^{-1/2}\exp\left[ -\frac{\|\bmy^{(i)}-\bmx^{(i)}\|^2}{2T}\right]\indi_{\bmy^{(1:m)}\in\mcal{H}} \label{eq:wiener}
\end{equation}
for $\bmx^{(i)}, \bmy^{(i)}\in\mbb{R}^d$. The proposal distribution (\ref{eq:wiener}) looks like a unit-drift Brownian motion of time length $T$ with the starting points drawn from $f_i, \,i=1,\dots,m$, and ending on the constraint.

\begin{lemma}\label{theorem1}
Under Condition \ref{condition:regularitycondition} in Appendix \ref{appx:proof1}, 
define
\begin{equation}\label{eq:phi}
\phi_i(\bm{u}) := \frac{1}{2}\left[\|\bmD \log f_i(\bm{u})\|^2+ \bmD\cdot\bmD \log f_i(\bm{u})\right] - l_i \geq 0, 
\end{equation} 
for some constant $l_i$ and $\bmD\cdot$ is the divergence operator (as opposed to gradient operator $\bmD$).
The transition density from $\bmx^{(i)}$ at time $0$ to $\bmy^{(i)}$ at time $T$ for the diffusion process (\ref{eq:langevin_process}) is given by
\begin{equation}
p_i(\bm y^{(i)}| {\bm x^{(i)}}) = 
\frac{f_i({\bm y}^{(i)})}{f_i({\bm x}^{(i)})} \cdot \left(\frac{1}{\sqrt{2\pi T}}\right)^d \exp\left( - \frac{\|\bmy^{(i)} -  \bmx^{(i)}\|^2}{2T}\right) \cdot \mathbb E \left[ \exp \left(- \int_0^{T} \left( \phi_i({\bm{\omega}}^{(i)}_s) + l_i \right)\dd s\right)\right] \label{eq:transition}
\end{equation}
and thus if we \textbf{disregard} the constraint $\mcal{H}$ for now,
\begin{equation*}
 \frac{g\left(\bm{x}^{(1)},\cdots, \bm{x}^{(m)}, \bm{y}^{(1)},\cdots,\bm{y}^{(m)}\right)}{h\left({\bm x}^{(1)},\cdots, {\bm x}^{(m)}, \bmy^{(1)},\cdots,\bmy^{(m)}\right)}  \propto  \mathbb E\left[\exp\left( - \sum_{i=1}^m \int_0^{T} \phi_i({\bm{\omega}}_s^{(i)}) \dd s \right)\right] 
\end{equation*}
where $\mathbb E$ is taking expectation over the measure induced by Brownian bridges $\bm{\omega}^{(1:m)}$ of length $T$ connecting $(\bmx^{(1)},\cdots,\bmx^{(m)})$ and $(\bmy^{(1)},\cdots,\bmy^{(m)})$.
\end{lemma}
The proof of this lemma is provided in Appendix \ref{appx:proof Lemma 2}.
The above Radon-Nikodym derivative \textbf{stays the same} under certain conditions after we include the constraints, i.e.,
\begin{corollary}\label{cor:constrained_rnd}
    Let $g_{\mcal{H}}\left(\bmx^{(1)},\dots,\bmx^{(m)},\bmy^{(1)},\dots,\bmy^{(m)}\right)$ given in (\ref{eq:biased_langevin}) and $h_{\mcal{H}}\left(\bmx^{(1)},\dots,\bmx^{(m)},\bmy^{(1)},\dots,\bmy^{(m)}\right)$ in (\ref{eq:wiener}).
    Suppose $\mcal{H}$ is a smooth manifold, then on the domain $\mbb{R}^{md}\times\mcal{H}$, and $g_\mcal{H}$, $h_{\mcal{H}}$ are integrable with respect to the product Lebesgue measure $\lambda_{\mbb{R}^{md}}\otimes \lambda_{\mcal{H}}$, then
\begin{equation}
 \frac{g_{\mcal{H}}\left(\bm{x}^{(1)},\cdots, \bm{x}^{(m)}, \bm{y}^{(1)},\cdots,\bm{y}^{(m)}\right)}{h_{\mcal{H}}\left({\bm x}^{(1)},\cdots, {\bm x}^{(m)}, \bmy^{(1)},\cdots,\bmy^{(m)}\right)}  \propto  \mathbb E\left[\exp\left( - \sum_{i=1}^m \int_0^{T} \phi_i({\bm{\omega}}_s^{(i)}) \dd s \right)\right] \label{eq:theorem1}
\end{equation}
where $\mathbb E$ is taking expectation over the measure induced by Brownian bridges $\bm{\omega}^{(1:m)}$ of length $T$ connecting $(\bmx^{(1)},\cdots,\bmx^{(m)})$ and $(\bmy^{(1)},\cdots,\bmy^{(m)})$, and $\phi_i$ as defined in (\ref{eq:phi}).
\end{corollary}
The proof of this corollary is provided in Appendix \ref{appx:proof Lemma 2}. Although (\ref{eq:theorem1}) is intractable, it is possible to construct a rejection sampling procedure that has the acceptance probability given by the right-hand side of (\ref{eq:theorem1}).
The procedure is sketched below and discussed in Appendix \ref{appx:poisson_process}.
\begin{remark}
The rejection stage can be done without computing the integral by simulating a Poisson point process on the space $[0,T]\times[0,M^{(i)}]$ for each $i$ where $M^{(i)}$ is an upper bound for the function $\phi_i$ and asserting if no point lies below the curve $\phi_i(\omega^{(i)}_s), s\in[0,T]$.
Provided that the functions $\phi_i$ are bounded above, this step is easy to execute (see Appendix \ref{appx:poisson_process}).
However, the function $\phi_i$ is usually not bounded above, in which case, one needs to determine the bounds for the proposal Brownian bridge and simulate the Brownian bridge conditioned on the pre-determined interval.
This approach is referred to as the "Layered approach for Brownian Bridge" by \cite{beskos2008factorisation}.
To avoid a complete re-iteration of the said paper, we will summarize the key steps in the appendix only and refer the reader to \cite{beskos2008factorisation} for the full detail.
\end{remark}

\subsection{Sampling from Constrained Proposals}\label{sec:proposal_sampling}
The preceding results ensure that if we can simulate the Brownian bridge that lands on the constraint $\mcal{H}$, a rejection step may be applied to correct the proposal into a sample from the constrained target distribution.
In other words, the problem is transformed from simulating an arbitrary distribution on an arbitrary manifold into simulating a Gaussian distribution on an arbitrary manifold.
Recall the constrained proposal distribution (\ref{eq:wiener})
\begin{equation}
 h_{\mcal{H}} \left(\bm{x}^{(1:m)},\bm{y}^{(1:m)}\right) \propto \left(\prod_{i=1}^m f_i(\bmx^{(i)})\right)\, f_{\bmy|\bmx}\left(\bmy^{(1:m)}|\bmx^{(1:m)}\right)\indi_{\bmy^{(1:m)}\in\mcal{H}}, \label{eq:wiener_c}
\end{equation}
where 
$$
f_{\bmy|\bmx}\left(\bmy^{(1:m)}|\bmx^{(1:m)}\right) := \prod_{i=1}^m(2\pi T)^{-1/2}\exp\left[ -\frac{\|\bmy^{(i)}-\bmx^{(i)}\|^2}{2T}\right].
$$
There are two possible ways one may handle the proposal:
\begin{enumerate}
    \item When $f_{\bmy|\bmx}\left(\bmy^{(1:m)}|\bmx^{(1:m)}\right)\indi_{\bmy^{(1:m)}\in\mcal{H}}$ can be directly sampled from with a tractable normalizing constant, then
    $$
    h_{\mcal{H}}\left(\bm{x}^{(1:m)},\bm{y}^{(1:m)}\right) \propto \underbrace{\left(\prod_{i=1}^m f_i(\bmx^{(i)})\right)}_{\text{Sample } \bmx^{(1:m)}} 
 \underbrace{Z_{\mcal{H}}(\bmx^{(1:m)})}_{\text{Accept/reject}} \underbrace{\frac{1}{Z_{\mcal{H}}(\bmx^{(1:m)})}f_{\bmy|\bmx}\left(\bmy^{(1:m)}|\bmx^{(1:m)}\right)\indi_{\bmy^{(1:m)}\in\mcal{H}}}_{\text{Sample } \bmy^{(1:m)} | \bmx^{(1:m)}}. 
    $$
    In this case, we can directly sample from the constrained Gaussian distribution and add a rejection step to correct for the normalizing constant.
    We showcase two types of constraints that may be treated this way in the Appendix \ref{appx:linear_constraint} and \ref{appx:spherical_constraint}. 

    \item In most other cases, it might not be trivial to sample from the constrained Gaussian distribution, or the normalizing constant is analytically intractable, then
    $$
    h_{\mcal{H}} \left(\bm{x}^{(1:m)},\bm{y}^{(1:m)}\right) \propto \underbrace{\left(\prod_{i=1}^m f_i(\bmx^{(i)})\right)}_{\text{Sample } \bmx^{(1:m)}} 
 \underbrace{f_{\bmy|\bmx}\left(\bmy^{(1:m)}|\bmx^{(1:m)}\right)}_{\text{Accept/Reject}}\underbrace{\indi_{\bmy^{(1:m)}\in\mcal{H}}}_{\text{Sample $\bmy^{(1:m)}$ uniformly from }\mcal{H}}. 
    $$
    Assuming we can sample uniformly from $\mcal{H}$, we can then obtain an exact sample from the target distribution (\ref{eq:cons1}).
\end{enumerate}
We summarize the sampling algorithm (named as Constrained Fusion Sampler) in Algorithm \ref{alg:CMCF-1} for the first case. 
The algorithm for the second case is given in the appendix.

\begin{algorithm}[t]
 \SetAlgoLined
 \SetKwInOut{Input}{input}
\Input{Manifold Constraint $\mcal{H}$; component distributions $f_i,i=1,\dots,C$; parameter $T$}
 Simulate, for each $1\leq i\leq m$, $\bm{x}^{(i)}\sim f_i(\cdot)$ \;
 Simulate $\bm{y}=(\bmy^{(1)},\dots,\bmy^{(m)})\sim \mathcal{N}(\bm{x}^{(1:m)}, T \boldsymbol I_d)$ constrained on $\bmy\in\mcal{H}$ \label{step:constraint Gaussian}\;
 Simulate a uniform random variable $U_1\in\mcal{U}[0,1]$\;
 \uIf{$U_1\leq Z_{\mcal{H}}\left(\bmx^{(1:m)}\right)$}{
    \For{$i=1,...,m$}{
        Simulate a Brownian Bridge of length $T$ connecting $\bm{x}^{(i)}$ and $\bm{y}^{(i)}$\;
    }
    Let $U_2\in\mcal{U}[0,1]$ and simulate the event ${\cal I}$ given by expression (\ref{eq:theorem1}), see Appendix \ref{appx:poisson_process}\;
    \uIf{$\mcal{I}$ is true}
    {
        Accept and return $\bm{y}^{(1:m)}$\;
    }
    \Else{
        Go back to step 1\;
    }
 }
 \Else {
    Go back to step 1\;
}
\caption{Constrained Fusion Sampler for Case 1}
\label{alg:CMCF-1}
\end{algorithm}

\subsection{Comparison with Existing Methods}\label{sec:comparison}

The Constrained Fusion algorithm proposed is based on a Gaussian proposal distribution.
We benefit from the proposal due to the ability to generate proposal samples distributed on the desired linear hyperplane.
Since the Gaussian proposal can be directly implemented into a naive importance sampling without the additional need to simulate diffusion processes, some may wonder how the proposed algorithm (Algorithm \ref{alg:CMCF-1}) performs compared with 
some other common variants of constrained sampler.
In this section, we conduct simulation studies to compare the algorithm performance in computing Monte Carlo estimates for linearly constrained models.
We considered the problem of computing the mean and variance of three independent random variables subject to a single sum constraint.
Four methods are tested on two different distributions and linear constraints. The results are shown in Fig. \ref{fig:error_compare}. Among these four methods, two of them are suitable for nonlinear constraints
are also tested on variance constraint, and results are shown in Figure \ref{fig:nonlinear_compare}. 
To compare the time efficiency, the computation time for each simulation per $10^4$ effective samples is plotted in Figure \ref{fig:time_comparison}.

Let $X_1$, $X_2$ and $X_3$ be three independent random variables subject to the constraint that $X_1+X_2+X_3=s$, where $s$ is known.
The distributions of $X_i$ are known and $n$ samples from the constrained joint distribution $\prod_i f_i(X_i)\indi_{X_1+X_2+X_3=s}$.
The following four constrained samplers are applied:
\begin{enumerate}
    \item The constrained fusion algorithm (Alg. \ref{alg:CMCF-1}), (\textbf{CF})
    \item The naive Gaussian proposal importance sampler. The proposal distribution is constructed by moment fitting the distributions of $X_j$ to generate three Gaussian approximations and then imposing the sum constraint. (\textbf{IS})
    \item The random-walk Metropolis-Hastings sampler. The sampler is initialized to a random point on the constraint hyperplane. At each step, the random walk displacement is drawn from a standard multivariate Gaussian distribution subject to the constraint that the dimensions sum up to zero, i.e., the sum of the components is unchanged so the constraint is still satisfied. The first $10^4$ samples are discarded. 
    (\textbf{MH})
    \item The Constrained Hamiltonian Monte Carlo (CHMC) as described in \cite{brubaker2012family}. Like in the Metropolis-Hastings case, the first $10^4$ samples are discarded. 
    The mass matrix is chosen to be the identity.
    (\textbf{CHMC})
\end{enumerate}
The $N$ drawn samples are used to compute the Monte Carlo approximated means and variances under the sum constraint, denoted as $\mu_i^N$ and $\Var_i^N$ respectively. 
The percentage error for estimated mean and variance at sample size $N$ are given by
\begin{equation}\label{eq:percentage_error}
    \text{PE}_{\mu_i}^N = \frac{\left|\mu_i^N-\mbb{E}[X_i]\right|}{|\mbb{E}[X_i]|}\times 100\%, \quad 
 \text{PE}_{\Var_i}^N=\frac{\left|\Var_i^N-\Var(X_i)\right|}{|\Var(X_i)|}\times 100\%, \quad i\in\{1,2,3\},
\end{equation}
where the ground truth $\mbb{E}[X_i]$ and $\Var(X_i)$ are computed using numerical integration, done using the 2-d integral functionality provided by the \emph{rmutil} package \citep{rmutil} in R \citep{rstats}.
The percentage errors in all three components are summed 
\begin{equation}\label{eq:percentage_error_total}
   \text{PE}_{\text{mean}}^N = \sum_{i=1}^3 \text{PE}_{\mu_i}^N, \qquad \text{PE}_{\text{Var}}^N = \sum_{i=1}^3 \text{PE}_{\Var_i}^N, 
\end{equation}
and the values are plotted against sample size $N$ in Fig. \ref{fig:error_compare}.

The computation time per $10^4$ effective samples for the subsequent simulation studies is summarised in Figure \ref{fig:time_comparison}. 
The effective sample size for the Markov chain-based methods is approximated by averaging the effective sample size of each output dimension using the \emph{ess} function provided in \emph{mcmcse} package \citep{mcmcse} in R.
From Figure \ref{fig:time_comparison}, we note that the CF algorithm requires high computational cost due to being a rejection algorithm, especially for the non-linear case. For the Linear cases, CF has a better computation efficiency compared with CHMC, since CF produces i.i.d. samples. 
For the non-linear case, the CF algorithm is able to explore the state space very well with only $600$ samples.
However, due to the extremely high rejection rate, the CF algorithm is essentially unviable to produce $10^4$ samples. A possible approach to improve the efficiency of the CF algorithm is discussed in Section 5.

Under all simulation scenarios, the tuning parameters of each algorithm are chosen to make sure that all algorithms reach a near optimal situation. For example, the random walk Metropolis-Hastings (MH) algorithm has a 42\% acceptance probability, and the importance sampling (IS) has an effective sample size of 25\% times the total particle size.
The constrained Hamiltonian Monte Carlo (CHMC) algorithm has roughly $80\%$ acceptance rate which is a good balance between computation time and jump size.
The constraint fusion (CF) algorithm uses an appropriate value of $T$ to obtain a high acceptance probability.  

\begin{figure}
    \centering
    \begin{subfigure}{0.76\linewidth}
        \centering
        \caption{}
        \includegraphics[width=\linewidth]{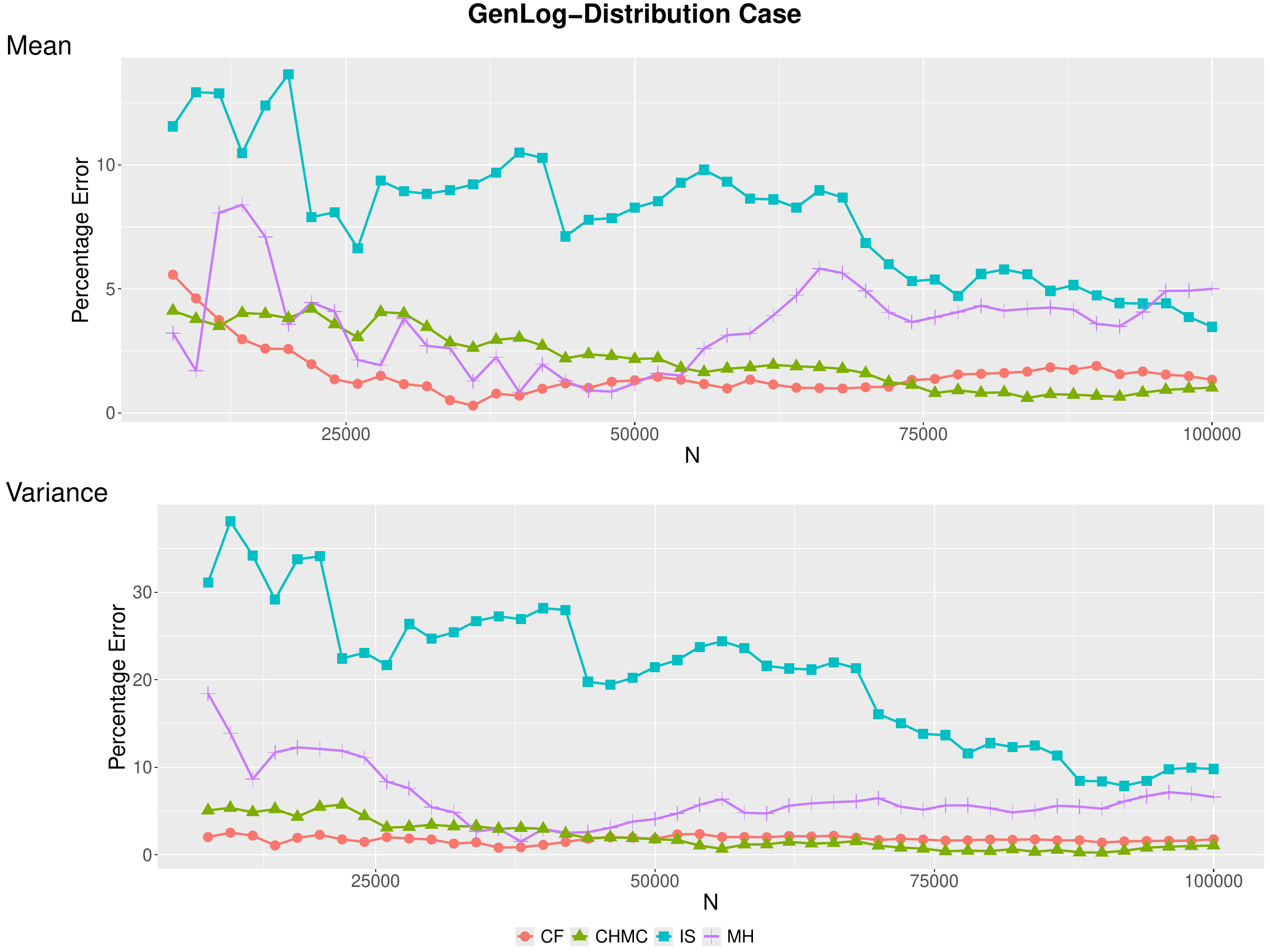}
        \label{fig:error_genlog}
    \end{subfigure}
    \begin{subfigure}{0.76\linewidth}
        \centering
        \caption{}
        \includegraphics[width=\linewidth]{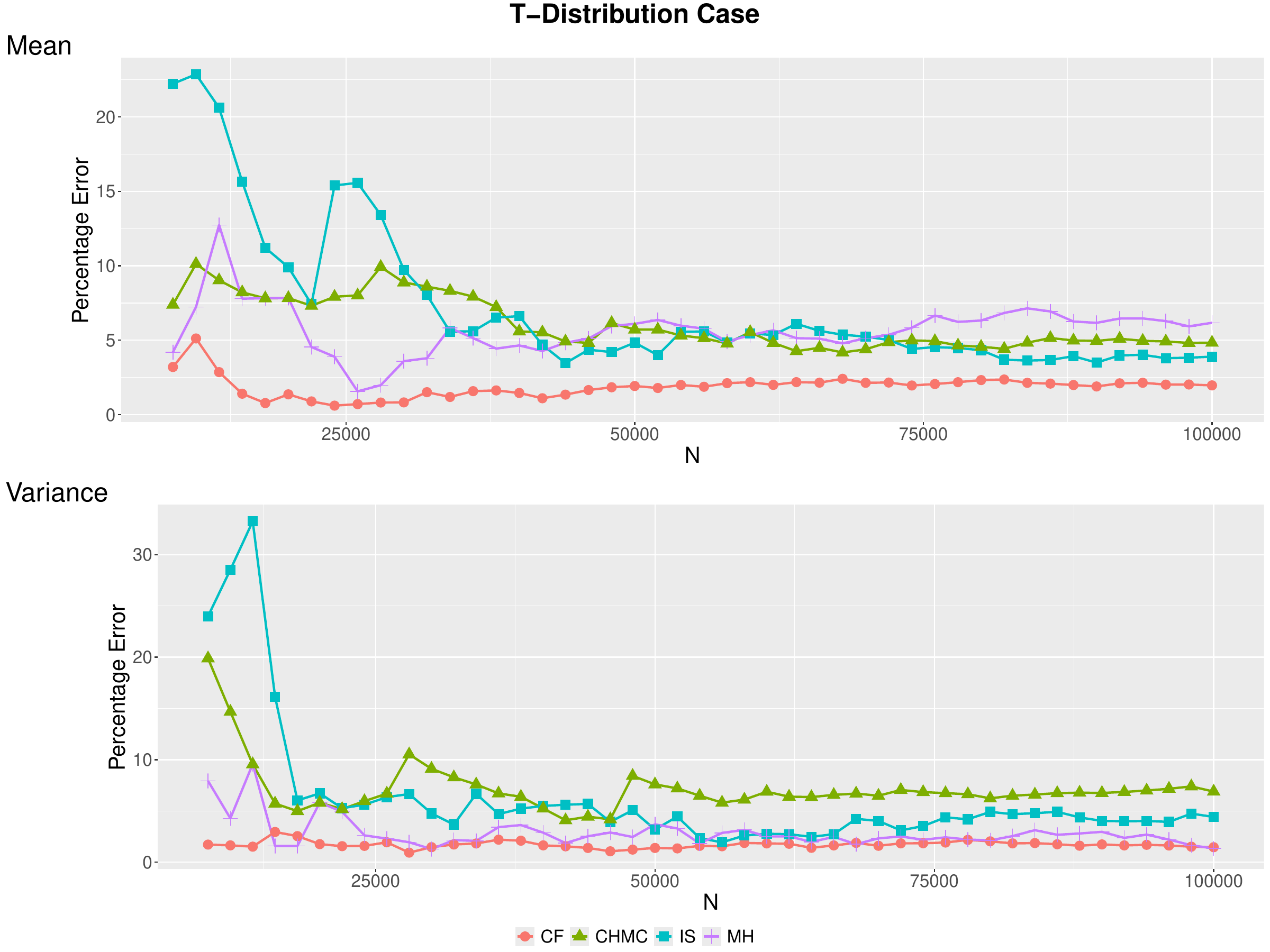}
        \label{fig:error_t}
    \end{subfigure}
    \caption{The percentage error curve varying with sample size $N$ is plotted for the mean and variance estimations (a) on the Generalized Logistic case and (b) on the T-distribution case.}
    \label{fig:error_compare}
\end{figure}

\subsubsection*{Case: Generalized Logistic Distribution}
Fig. \ref{fig:error_genlog} considers the Generalized Logistic distributions as described in \cite{halliwell2018log}, see Appendix \ref{appx:genlog} for more detail.
\begin{definition}[Generalized Logistic Distribution]
    Let $\alpha,\beta,\gamma>0$, $C\in\mbb{R}$, and $Y_1\sim\Gamma(\alpha,1)$, $Y_2\sim\Gamma(\beta,1)$.
    Let 
    $$X:=\gamma \log\left(\frac{Y_1}{Y_2}\right) + C,$$
    then $X$ is said to follow a Generalized Logistic distribution with parameter $(\alpha,\beta,\gamma,C)$, denoted $X\sim \GenLog(\alpha,\beta,\gamma,C)$.
\end{definition}
The Generalized Logistic distribution can model both positively and negatively skewed data.
The distribution always has a heavier tail than the normal distribution.

The setup assumes that 
$X_1\sim\GenLog(3,0.4,2,-5)$, $X_2\sim\GenLog(3,0.4,1,-2)$ and $X_3\sim\GenLog(3,0.4,1,-3)$.
The sum constraint is 
$$
\mcal{H}:=\{(X_1,X_2,X_3)\in\mbb{R}^3 : X_1+X_2+X_3=10\}.
$$
The random variables are positively skewed and leptokurtic while the mean sum is about $-6$ away from the sum constraint.
The performance is evaluated using the total percentage error given in (\ref{eq:percentage_error_total}) and error is plotted against the number of samples in Fig. \ref{fig:error_genlog}.
We can see that the CHMC algorithm and our CF algorithm quickly converged (to below $5\%$ error) while the other two algorithms would take a bit longer.

\subsubsection*{Case: Student's T-Distribution}
Simulation results in Fig. \ref{fig:error_t} are conducted on mean-shifted T-distribution, to test against a heavier-tailed distribution.
Use $X\sim T_{\nu}(\mu)$ to denote the random variable $X:=Y+\mu$, where $Y\sim T(\nu)$ is a standard T-distribution with degrees of freedom $\nu$.
The setup assumes that $X_1\sim T_{2.01}(-2)$, $X_2\sim T_{2.01}(3)$ and $X_3\sim T_{2.01}(5)$.
Again, the sum constraint is set to be 
$$
\mcal{H}:=\{(X_1,X_2,X_3)\in\mbb{R}^3 : X_1+X_2+X_3=10\}.
$$
The percentage error is computed using (\ref{eq:percentage_error_total}) and plotted against the number of samples in Fig. \ref{fig:error_t}.
Similar to the generalized logistic case, the CF estimates with faster convergence and are more stable over the run.

\begin{remark}
Both the random-walk MH sampler and the CHMC sampler require the user to manually decide some parameters, e.g., step-size, mass matrix, etc. 
The choice of these parameters directly links with the convergence and estimation accuracy of the sampler but are usually not trivial to choose.
In contrast, the CF sampler only has one tuning parameter $T$ which only affects the efficiency but not the accuracy.
\end{remark}

\subsubsection*{Case: Non-linear Constraint}
In this example, we consider sampling from a product T-distribution constrained on the sample mean and variance given by
$$
    f_{\mcal{H}}(X_1,X_2,X_3) := f_{T}(X_1;0, 0.6, 9) f_{T}(X_2;0,0.6,9) f_{T}(X_3; 0, 4, 3) \indi_{(X_1,X_2,X_3)\in \mcal{H}} 
$$
where $f_{T}(\cdot; \mu,\sigma, \nu)$ is the density function of a non-central T-distribution with mean $\mu$, scale $\sigma$ and degree of freedom $\nu$ given by
$$
f_{T}(x;\mu,\sigma,\nu):= \frac{\Gamma\left(\tfrac{\nu+1}{2}\right)}{\sqrt{\pi\nu}\Gamma(\nu/2)\sigma}\left(1+\frac{(x-\mu)^2}{\sigma^2}\right)^{-\frac{\nu+1}{2}},
$$
and
$$
\mcal{H}:=\left\{(x_1,x_2,x_3): \sum_{i}x_i=0, \frac{1}{3}\sum_{i}x_i^2=8\right\}.
$$
The unconstrained target distribution is designed to be bi-polar where the density in the first two dimensions is concentrated around $0$ but the third dimension is allowed to take a range of values.
Such a density, subject to mean and variance constraints, gives rise to a circle in 3D with at least two modes sitting opposite to each other.
Indeed, looking at Fig. \ref{fig:nonlinear_compare}, we see that the target actually has four modes, divided into two clusters.
Due to the multimodality, CHMC (in Fig. \ref{fig:chmc_non_linear}) with 10000 samples failed to explore the whole space and only produced samples from the lower half of the space.
In contrast, the Constrained Fusion (in Fig. \ref{fig:fusion_non_linear}) with only 600 samples already gives a good representation of the target distribution with all four modes discovered.

\begin{figure}[tb]
\centering
    \begin{subfigure}{0.45\linewidth}
        \centering
        \caption{Constrained Fusion}
        \includegraphics[width=\linewidth]{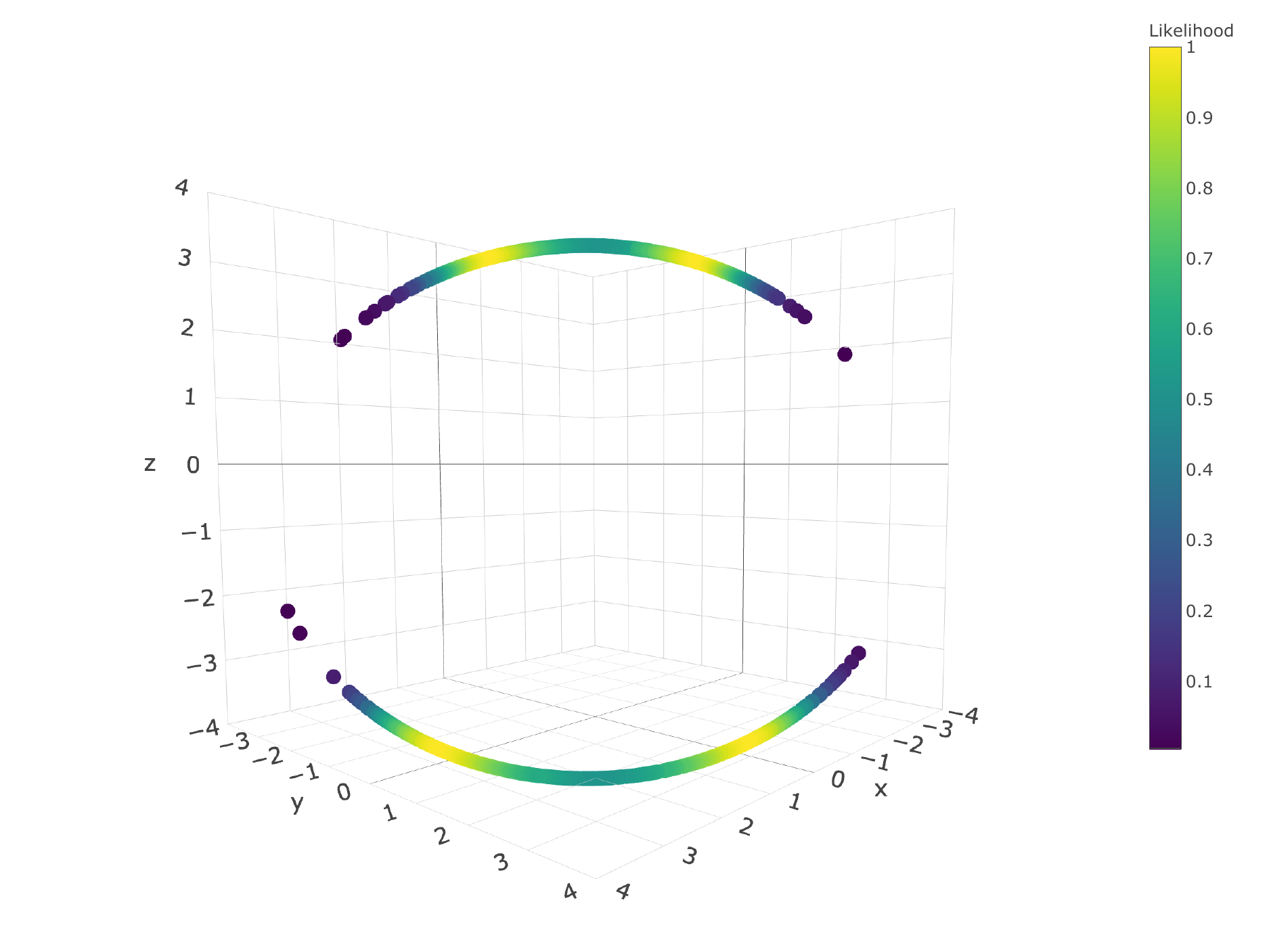}
        \label{fig:fusion_non_linear}
    \end{subfigure}
    \begin{subfigure}{0.45\linewidth}
        \centering
        \caption{CHMC}
        \includegraphics[width=\linewidth]{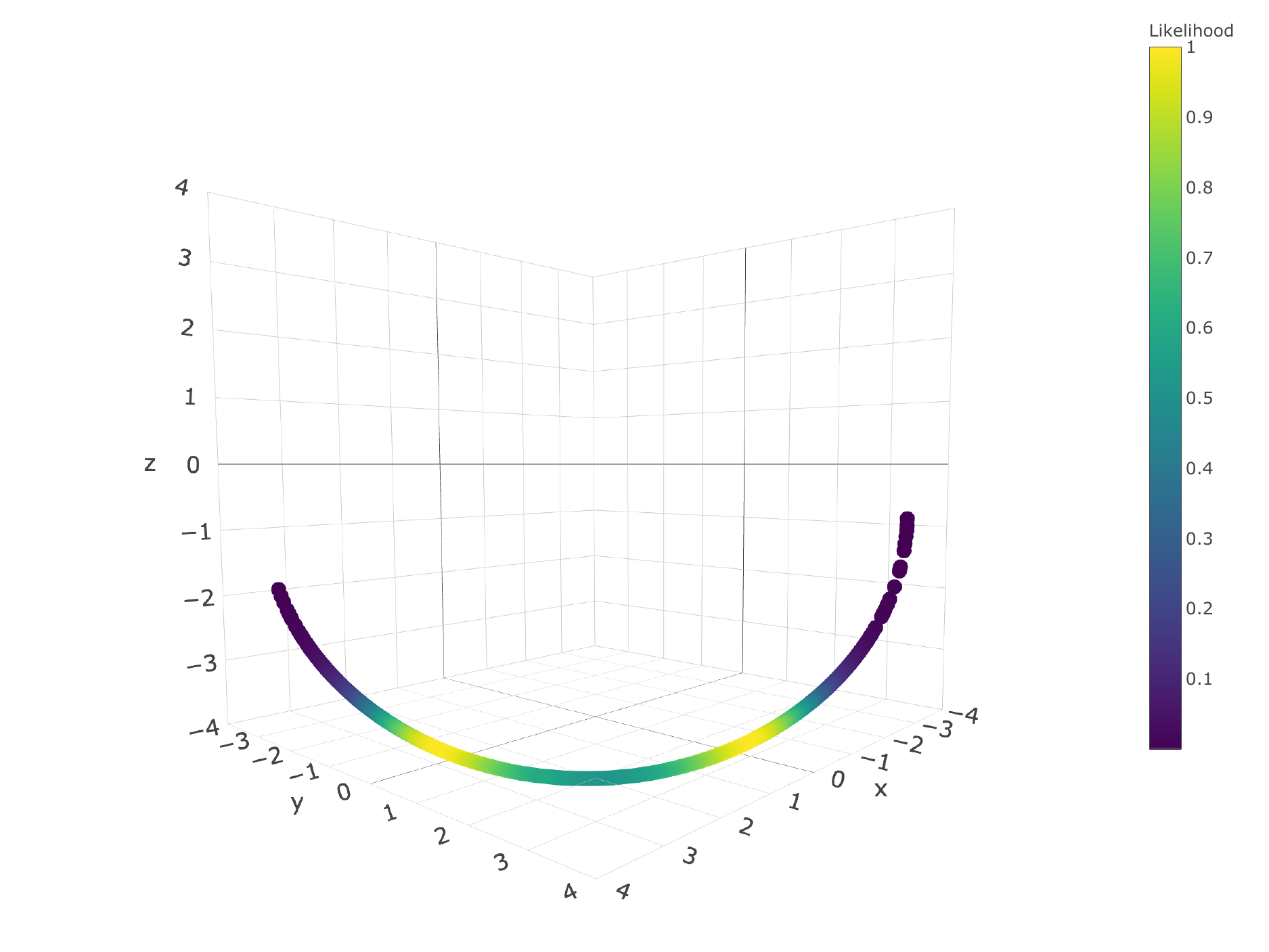}
        \label{fig:chmc_non_linear}
    \end{subfigure}
    \caption{Drawn samples using Constrained Fusion (left) and CHMC (right) plotted, the colour indicates the un-normalized likelihood value of the sample. The four modes of the constrained distribution are shown in yellow in the above plots. The CMHC (right) failed to find all four modes.}
    \label{fig:nonlinear_compare}
\end{figure}

\begin{remark}
    Note that CHMC is used in both cases with the same tuning parameter, standalone or as the backend for constrained fusion to generate uniform points. 
    The difference is that for Fig. \ref{fig:chmc_non_linear}, the sampler runs on the manifold with a non-uniform potential, whereas for Fig. \ref{fig:fusion_non_linear}, the CHMC algorithm only needs to produce samples from a uniform distribution on the constraint and the hard-lifting is all done by the Fusion algorithm.
    In this case, it is much easier to generate samples uniformly from the constraint since the sampler doesn't need to traverse a multimodal terrain.
\end{remark}

\begin{figure}[t]
    \centering
    \caption{Computation time (sec) for each simulation in Section 3.3}
    \includegraphics[width=0.9\linewidth]{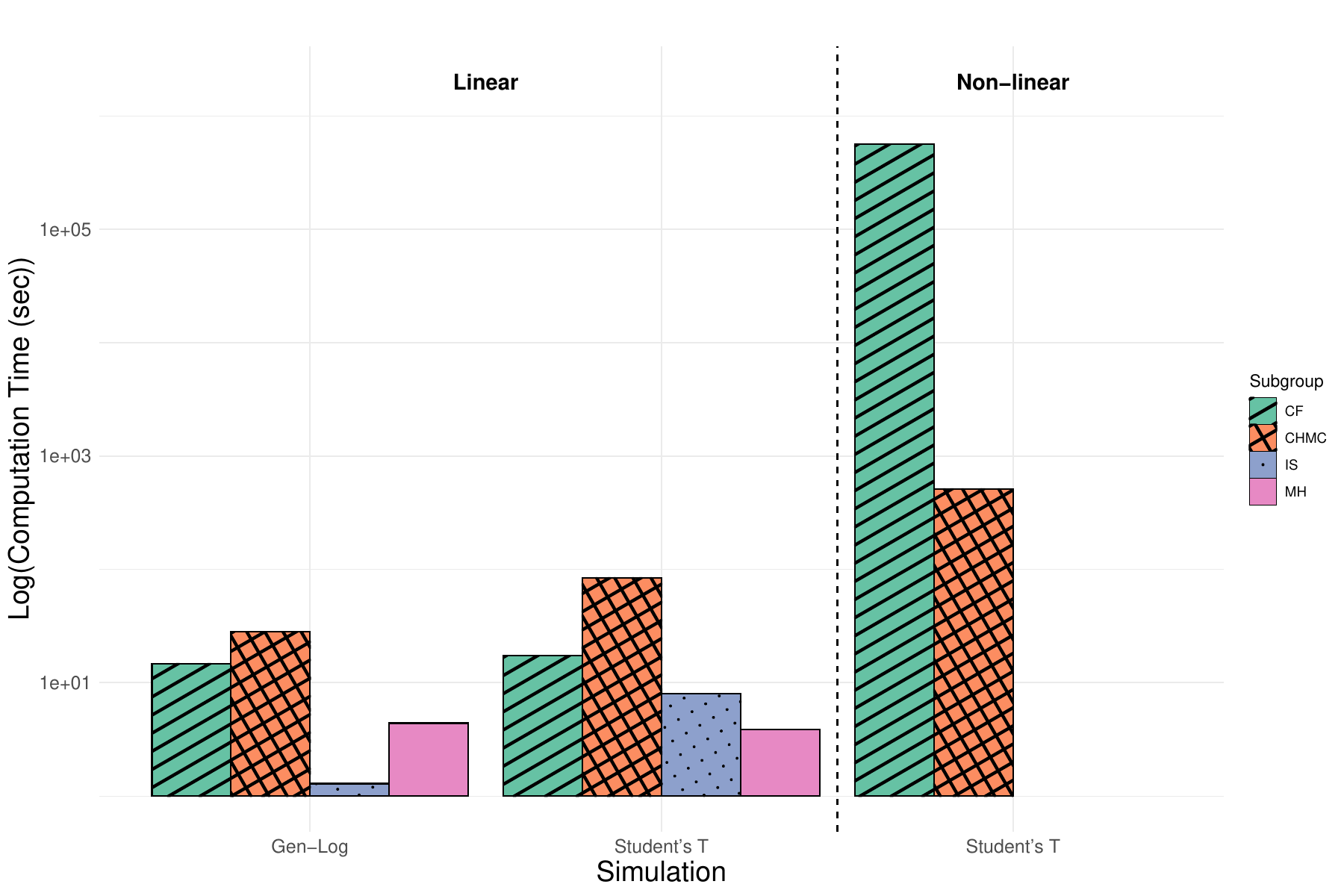}
    \label{fig:time_comparison}
\end{figure}


\newpage

\section{Application to Time Series Imputation}\label{sec:application_framework}
In this section, we focus on problems related to frequency up-scaling or disaggregation of time series models subject to a constraint.
Such problems often arise in consumption modeling, e.g., power consumption, water usage, etc., where the consumption data are recorded for a site at a certain frequency.
However, due to cost, privacy or other considerations, the observer, e.g., the energy provider, might opt to later record at a lower frequency than before, for instance, from 10 times per hour to twice per hour.
It is natural to question if one can disaggregate the later low-frequency readings and recover high-frequency readings from them.
Here, the low-frequency data poses linear constraints in the imputation problem, since the estimated high-frequency consumption should sum to the observed reading {at a lower frequency}.
The Constrained Fusion sampling algorithm presented above can be applied to simulate high-frequency readings that satisfy the constraint while preserving the statistical properties of the model. This section introduces the basic time series model in Section \ref{sec:model} and then we study a typical scenario in time series disaggregation with respect to either linear or non-linear constraints in Section \ref{sec:study1}. We showcase two other applications of constrained simulation in Appendix \ref{appx:application}.

\subsection{Basic Model}\label{sec:model}
Consider two parallel time series $\{S_{t}\}_{t=1,2,3,\dots}$ and $\{Y_{t}^{(i)}\}_{t=1,2,3,\dots}^{i=1,2,\dots,m}$ where $S_t$ is the low-frequency data and $Y_{t}^{(i)}$ is recorded $m$ times as frequently as $S_t$.
Temporally, the recordings are taken in this order:
$$
\cdots, \underbrace{(S_{t-1}, Y_{t-1}^{(m)})}_{\text{simultaneous}}, Y_t^{(1)}, Y_t^{(2)},\cdots, Y_t^{(m-1)}, (S_t, Y_t^{(m)}), Y_{t+1}^{(1)},\cdots,
$$
In particular, the recordings are such that $\sum_{i=1}^m Y_t^{(i)}=S_t \ \forall t$, since each time series records total consumption within the time period considered.
For the rest of this section, rather than taking $Y^{(i)}_t$ as a single time series, we would treat $\{Y_t^{(1)}\}_{t=1,2,\dots}$ to $\{Y_t^{(m)}\}_{t=1,2,\dots}$ as independent and model them separately.
\begin{remark}
The time series $Y_t^{(i)}$, $i=1,\dots,m$ are modelled independently to fit with the form of (\ref{eq:cons1}) where each $Y_t^{(i)}$ corresponds to a separate density.
However, it is possible to introduce dependency into $(Y_t^{(1)},\dots,Y_t^{(m)})$ when the dependency cannot be disregarded.
For instance, we considered a copula structure on the random variables in Example \ref{sec:study2} in the Appendix.
\end{remark}


\subsubsection{Autoregressive Model}
In general, we will use the Autoregressive (AR) model to fit each high-frequency time series $Y_t^{(i)}$, $i=1,\dots,m$.
In the same experiment, the $m$ AR models will share the same model structure, e.g., model order $K$ and choice of additional regressors, but can have distinct model parameters.
To estimate the parameters, we use the full resolution measurements $Y_t^{(i)}$ are usually available in practice for a certain past period of $t\in I$, which will form the training set later denoted as $(\hat{Y}_t^{(i)})$.

Referring back to Fig. \ref{fig:error_qq}, it is clear that a vanilla AR model $Y_{t}^{(i)}$ is not suitable for this data, since the model residuals are clearly not Gaussian.
In addition, energy consumption data are non-negative and positively skewed, so one may need to turn towards some non-Gaussian distributions to model the error term $\epsilon_t^{(i)}$. 
Here we use Generalized Logistic distribution \citep{halliwell2018log} which accommodates positive skewness.
Now, the AR model (of order K) with a Generalized Logistic link is given by, for $i=1,\dots,m$,
\begin{eqnarray}\label{eq:ARmodel}
 Y_{t}^{(i)} \sim  \GenLog\left(\alpha^{(i)}, \beta^{(i)}, \gamma^{(i)}, C^{(i)} + \mu_t^{(i)}\right), \;\;\;\;\;\;
\mu_{t}^{(i)} = \sum_{r=1}^K \Phi_r^{(i)} Y_{t-r}^{(i)} + \bmXi_t^{\top} \bm{\psi}^{(i)}
\end{eqnarray}
where $\bmXi_t$ are additional covariates and the parameters $\alpha^{(i)}, \beta^{(i)}, \gamma^{(i)}, C^{(i)}, \bm{\Phi}_r^{(i)}, \bm{\psi}^{(i)}$ are unknown.

We take a straightforward approach to fitting the model parameters.
Firstly, the regression parameters $\bm{\Phi}^{(i)}$ and $\bm{\psi}^{(i)}$ are fitted using the least-squares method, disregarding the distribution of $Y_t^{(i)}$.
With $\bm{\Phi}^{(i)}$ and $\bm{\psi}^{(i)}$ held fixed, the residuals are used to fit the parameters for the generalized logistic distribution by moment (cumulant) fitting.

In detail, let $n$ denote the total number of low frequency time points used in training the model under the 
training set of observed data $(\hat{Y}_t^{(i)})_{t\in\{1,\dots,n\}, i\in\{1,\dots,m\}}$.
Then, we first fit $\bm{\Phi}^{(i)}=\left(\Phi_{1}^{(i)},\dots,\Phi_{K}^{(i)}\right)$, and $\bm{\psi}^{(i)}$ by minimising the empirical squared loss, i.e., for each $i\in\{1,\dots,m\}$, solve
\begin{equation*}
    \argmin_{\bm{\Phi}^{{(i)}},\bm{\psi}^{(i)}} \sum_{t=1}^{n} \left[ \hat{Y}_{t}^{(i)}-\hat{\mu}_{t}^{(i)}\right]^2, \quad \hat{\mu}_{t}^{(i)} := \sum_{r=1}^K \Phi_{r}^{(i)} \hat{Y}_{t-r}^{(i)} + \bmXi_{t}^{\top} \bm{\psi}^{(i)}.
\end{equation*}
Then 
by computing $\hat{\mu}_{t}^{(i)}$ using the fitted parameters $\hat{\bm{\Phi}}^{(i)},\psi^{(i)}$, we fit a Generalised Logistic distribution with parameters $(\alpha^{(i)},\beta^{(i)},\gamma^{(i)},C^{(i)})$ for each $i$ using the residues $\hat{Y}_t^{(i)}-\hat{\mu}_t^{(i)}$.
The parameters $\alpha^{(i)},\beta^{(i)},\gamma^{(i)}$ are fitted by matching
the variance, skewness and excess kurtosis of the residues (see Appendix \ref{appx:genlog} for more detail).
Finally, the parameter $C^{(i)}$ is computed such that $\GenLog(\alpha^{(i)},\beta^{(i)},\gamma^{(i)},C^{(i)})$ has mean $0$.

\subsubsection{Constrained Imputation}
Given the parameter estimates $\left(\hat{\bm{\Phi}}^{(1:m)}, \hat{\bm{\psi}}^{(1:m)}, \hat{\alpha}^{(1:m)}, \hat{\beta}^{(1:m)},\hat{\gamma}^{(1:m)},\hat{C}^{(1:m)}\right)$, which have been fitted using an initial set of high-frequency data, we now focus on utilizing this fitted model to impute later missing high-frequency data that is constrained by low-frequency observations.
Let $Y_t^{(1:m)}$, $t=1,2,\dots, \mcal{T}$ be the high-resolution energy consumption we want to impute with respect to low-resolution time series data $\{S_t\}_{t=1,\dots,\mcal{T}}$ which is observed. Note that the time indices here are \emph{not} the same as in the previous subsection, 
i.e., $t$ also counts from $1$ for the test set in the sense that the training set $\hat{Y}_t^{(i)}$ is now disregarded after fitting the parameters.
By the construction of the AR model, every time point $\bmY_t$ depends only on its previous times $\bmY_{t-1},\bmY_{t-2},\dots$, and the constraint $\bmY_t\in\mcal{H}_t$.
To illustrate the workflow, we will assume a sum constraint, $\sum_{i=1}^m Y_t^{(i)}=S_t$, for now, but the process is also applicable to non-linear constraints.
Then given the explanatory variables $\bmXi_t$ at time $t$, the density of $\bm Y_t$ conditioned on the past is
\begin{equation}
 f(\bmY_t|\bmY_{t-1:t-K}) = \prod_{i=1}^m f\left(Y_{t}^{(i)}|Y_{t-K:t-1}^{(i)}\right) \indi_{\sum_{i=1}^m Y_t^{(i)}=S_t} \label{eq:pred_like}
\end{equation}
where $Y_t^{(i)}\sim \GenLog(\alpha^{(i)}, \beta^{(i)},\gamma^{(i)}, C^{(i)}+\mu_t^{(i)})$, $\mu_t^{(i)}=\sum_{r=1}^K\Phi_r^{(i)}Y_{t-r}^{(i)}+\Xi_t^{\top} \bm{\psi}^{(i)}$
As we can see, for the simulation of each $\bmY_t$, the target distribution (\ref{eq:pred_like}) follows the shape of (\ref{eq:cons1}) where each factored density $f\left(Y_{t}^{(i)}|Y_{t-K:t-1}^{(i)}\right)$ can be easily simulated and the product is subject to a linear constraint.
Thus Algorithm \ref{alg:CMCF-1} can be applied directly to this simulation problem by applying it sequentially in temporal order (see Fig \ref{fig:time_series_imputation})
$$
\bmY_{1}^{(1:m)} \rightarrow \bmY_{2}^{(1:m)} \rightarrow \bmY_{3}^{(1:m)} \rightarrow \cdots \rightarrow \bmY_{\mcal{T}}^{(1:m)}.
$$
To simulate time point $\bmY_t$, one would first draw a sample $\bmx$ from (\ref{eq:pred_like}) without the constraint as the starting points of the Brownian bridge.
Then sample $\bmy\sim \mcal{N}(\bmx, T \mathbf I_k)$ subject to sum constraint $\|\bmy\|_1 = S_{t}$ which can be vectorized as $\bm{A}\bmy=S_t$, where $\bm A$ is a row matrix of ones
\footnote{Note this step is different for non-linear case.}.
Then the sampled particle goes through two rejection steps as described in Alg. \ref{alg:CMCF-1}.
In the case where nonlinear constraints are used, we use Alg. \ref{alg:CMCF-2} with CHMC \citep{lelievre2019hybrid} as the base constrained uniform sampler and follow Alg. \ref{alg:CMCF-2}.

\begin{figure}
    \centering
    \includegraphics[width=0.7\linewidth]{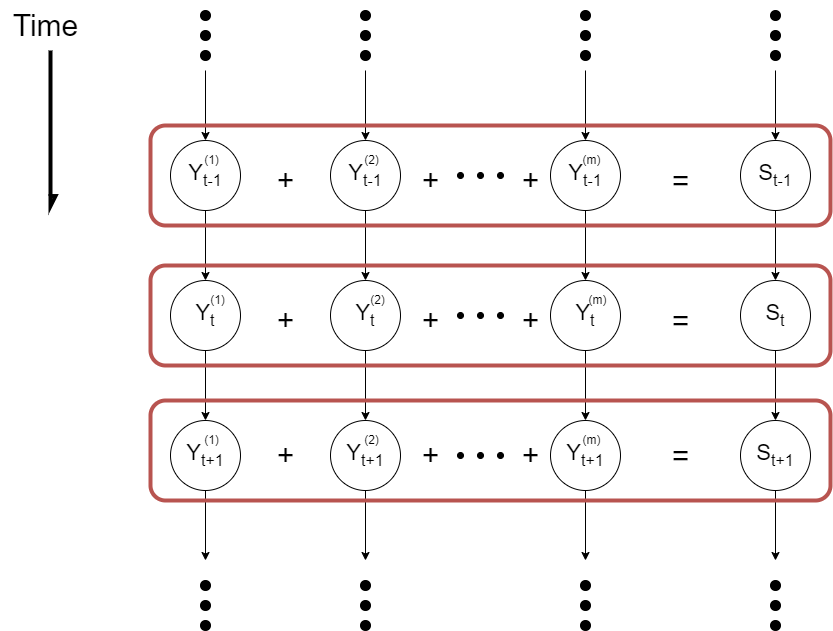}
    \caption{Constrained Imputation for Time series.}
    \label{fig:time_series_imputation}
\end{figure}

\subsection{Study 1: Day-readings Disaggregation} \label{sec:study1}
In this section, we consider a problem that electricity companies may encounter.  
Modern time-of-use meters are capable of reporting electricity usage for three periods per day (peak time, off-peak time, and midnight). 
However, these high-resolution meter readings may be missing due to unreliability or delay, etc., and on such days we only obtain one reading per day as for earlier generation meters.
We are interested in recovering the high-resolution energy consumption in each time period from the low-resolution aggregated consumption.

We use a subset of the data published in the Irish Smart Meter Trial \citep{CER1,CER2}, which includes half-hourly energy consumption readings of individual residential smart meters from July 2009 to the end of 2010 (in total 535 days) and corresponding questionnaire data of residential customers including the social-economic data of occupants.
From the original dataset, we randomly extracted 31 households that have no missing entries in the columns of survey responses of our interest.
These responses are used as covariates for the AR model, see Appendix \ref{appx:covariates} for the list of variables used.
{We processed the data to consider the problem of disaggregating the daily readings into thrice daily readings.}

\subsubsection*{Parameter Estimation}
In this particular problem, the aggregated readings form a time series with a unit of day and the goal is to impute a time series with three times the frequency, i.e., three readings per day.
Therefore, we will need three separate AR models to impute the time series.
Recall that the equation for each AR model is given by
\begin{equation}\label{eq:ARmodel1}
 Y_{j,t}^{(i)} \sim  \GenLog\left(\alpha^{(i)}, \beta^{(i)}, \gamma^{(i)}, C^{(i)}+\mu_{j,t}^{(i)}\right), \;\;\;\;\;\;
\mu_{j,t}^{(i)} = \sum_{r=1}^K \Phi_r^{(i)} Y_{j,t-r}^{(i)} + \bmXi_j^{\top} \bm{\psi}^{(i)}
\end{equation}
where $i\in\{1,2,3\}$ is the indexing for the three separate time series.

In fitting the model, the data may come from multiple customers and we use an additional subscript $j$ to denote the data from customer $j$ in equation (\ref{eq:ARmodel1}).
However, the parameters $\alpha^{(i)},\beta^{(i)},\gamma^{(i)}, C^{(i)},\bm{\Phi}^{(i)}, \bm{\psi}^{(i)}, \,i=1,2,3$ are assumed to be the same for all customers.
The model order $K$ is chosen to be $7$ as people tend to have their regular activities repeated weekly.

Entries from questionnaire data, including the number of adults/children in the household, number of bedrooms, number of large electrical appliances, etc., are used as additional covariates $\bmXi_j$ in this model.
The subscript $t$ has been dropped since the survey data is time-independent.
The high-frequency data of the extracted 30 households across the first 20 days are used to fit the model parameters.
The remaining one household's data is used for simulation.
As mentioned before, $\bm{\Phi}$ and $\bm{\psi}$ are fitted through least-squares estimation and $\alpha^{(i)},\beta^{(i)},\gamma^{(i)}, C^{(i)}$ from moment (cumulant) fitting.

\subsubsection*{Result}
\begin{figure}[t!]
 \centering
 \begin{subfigure}[t]{\linewidth}
    \centering
    \caption{Constrained on Mean}
    \includegraphics[width=0.8\linewidth]{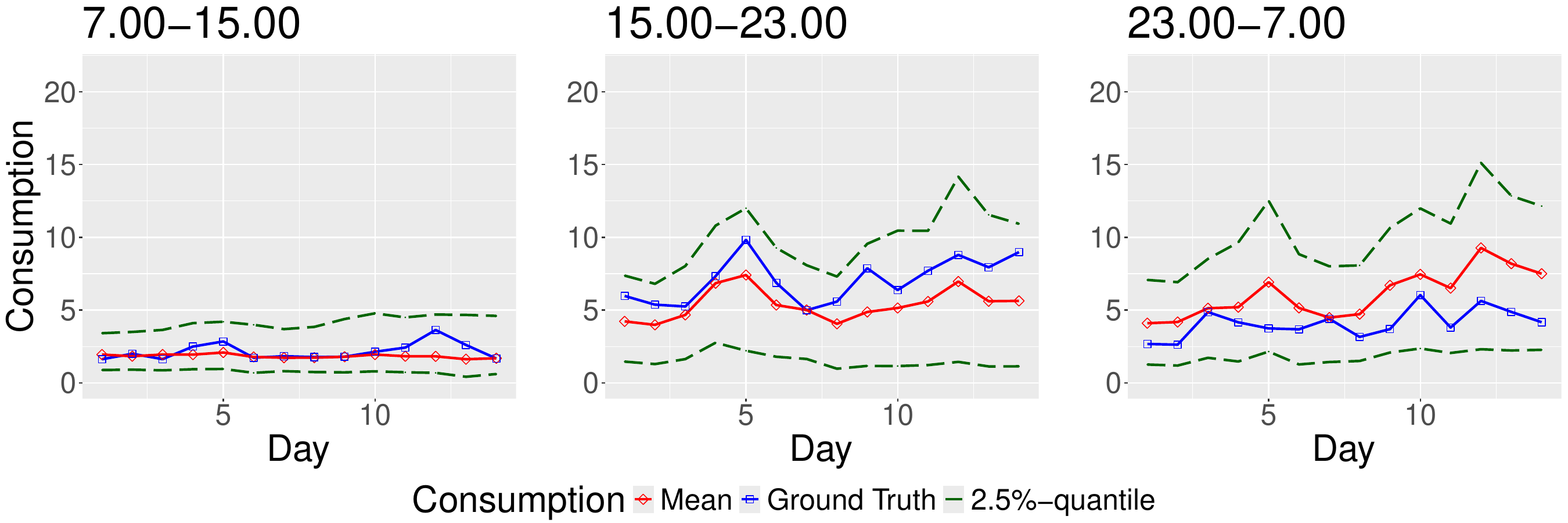}
     \label{fig:cons}
 \end{subfigure}

 \begin{subfigure}[t]{\linewidth}
    \centering
    \caption{Constrained on Mean and Variance}
    \includegraphics[width=0.8\linewidth]{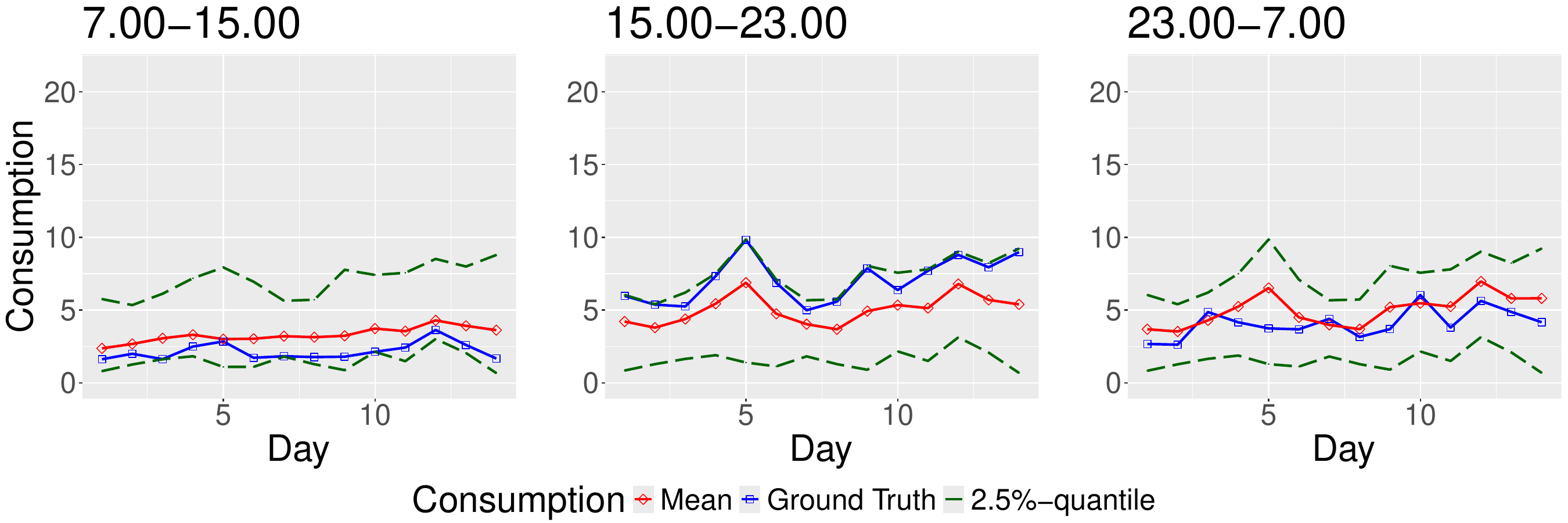}
     \label{fig:cons_var}
 \end{subfigure}
 
 \begin{subfigure}[t]{\linewidth}
    \centering
    \caption{Unconstrained simulation}
    \includegraphics[width=0.8\linewidth]{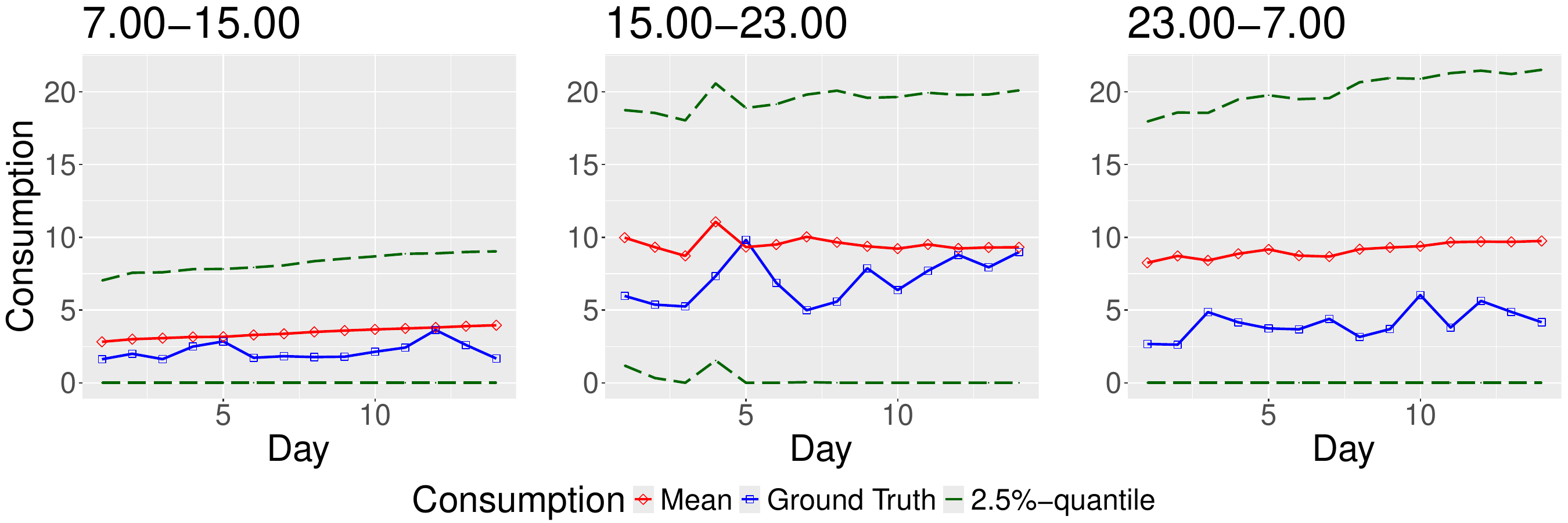}
     \label{fig:uncons}
 \end{subfigure}

 \caption{Energy consumption imputation and error with and without constraints}
 \label{fig:predictions}
\end{figure}

After estimating the set of parameters $\left(\hat{\bm{\Phi}}^{(1:3)},\hat{\bm{\psi}}^{(1:3)},\alpha^{(1:3)},\beta^{(1:3)},\gamma^{(1:3)}, C^{(1:3)}\right)$, we have three AR models of order 7 to estimate separately energy consumption in the three time periods of the day.
Taking the high-frequency data for the first 7 days in the dataset as the start of the time series, imputation is conducted on the remaining one household to up-scale the daily readings of the subsequent 14 days to a time series of length 52.
By implementing Algorithm \ref{alg:CMCF-1}, the imputation is done in temporal order progressing in time $t$, with $10^4$ samples of $Y_t^{(1:3)}$ drawn in each step to compute an estimate for the sample mean, sample variance and, the 95\% confidence interval.
Results are presented in Fig. \ref{fig:predictions}.
Three simulations are done with the same underlying model:
\begin{enumerate}
    \item Energy consumption imputed with respect to sum constraint 
    $$
    \mcal{H}_t:=\left\{Y_t^{(1:m)}:\sum_{i=1}^m Y_t^{(i)}=S_t\right\},
    $$
    shown in Fig. \ref{fig:cons};
    \item Energy consumption imputed with respect to sum and variance constraint
    $$
    \mcal{H}_t:=\left\{Y_t^{(1:m)}:\sum_{i=1}^m Y_t^{(i)}=S_t, \sum_{i=1}^m (Y_t^{(i)}-S_t)^2=\Sigma_t\right\},
    $$
    shown in Fig. \ref{fig:cons_var};
    \item Energy consumption estimated without any constraint, shown in Fig. \ref{fig:uncons}.
\end{enumerate}
In these figures, the sample mean, ground truth, and $95\%$ CI are plotted in red (diamond), blue (square) and green (dashed) lines respectively.

\noindent \textbf{Linear Constraint:}
Comparing the results under sum constraint (Fig. \ref{fig:cons}) and without constraint (Fig. \ref{fig:uncons}) we see directly the significance of injecting the information from the sum constraint.
The sample mean trajectory from the constrained model shows similar fluctuation as the real trajectory, unlike the unconstrained model where the estimated mean is mostly flat.
We also see improvements in the estimated $95\%$ CI in the constrained case, as the $95\%$-CI is always increasing with time for the unconstrained case.
This is reasonable since without extra information the uncertainty could only increase as extra uncertainty is injected every time step.
On the other hand, in Fig. \ref{fig:cons}, the constrained $95\%$-CI does not seem to increase with time.
Instead, it increases with the true value which is reasonable as the variance of Gamma distribution is proportional to its mean squared.

\noindent \textbf{Non-Linear Constraint:}
Comparing the sum constraint result (Fig. \ref{fig:cons}) and the sum-variance constraint (Fig. \ref{fig:cons_var}), we see the uncertainty estimation is tighter for the second and third time intervals where the electricity consumption is higher, but a more conservative uncertainty estimation when the consumption is low.
This is possibly due to the mean and variance constraint always drawing from a sphere.

One observation in the non-linear constraint case is that the ground truth lies very close to the quantile points at some time indices.
This is due to the fact that the non-linear constraint which restricts both the mean and variance of the imputed values is essentially restricting the sample space into a bounded set.
Thus when the ground truth value (of one coordinate) is close to the upper boundary of the constraint set, the $97.5\%$-quantile line is essentially as informative as the upper bound, and hence the coincidence.
The same applies when the ground truth is near the lower boundary of the constraint set.

\section{Discussion}\label{sec:discussion}
In this paper, we presented a novel sampling approach for constrained problems that could deal with a wide range of density functions and constraints, as long as the unconstrained density can be factored into a product of density functions.
The proposed method has been demonstrated, in the simulation studies, to have a better estimation efficiency than the naive counterparts.

For example, we demonstrated that in time series datasets that record accumulated values, e.g., energy consumption, it is natural to apply sum constraints if the data is available.
The resulting joint distribution will take the form of (\ref{eq:cons1}) as a product density subject to some constraint on the summary statistics. 
In those situations, our algorithm can be applied to obtain samples and estimates for further analysis.
We have shown the effectiveness of linearly constrained models in dealing with disaggregating time series data in Section \ref{sec:study1} with two other examples in Appendix \ref{appx:application}.
These situations have their significance in the domain of energy supply and retail, where the ability to impute high-resolution time series from a low-resolution time series could help the supplier to more accurately model the characteristics of energy consumption and potentially reduce cost in maintaining the monitoring system.
We also applied the algorithm to non-linear constraints where summary statistics of higher order are also obtained and used.

We have analyzed theoretically, in the Gaussian case, that adding a sum constraint to an unconstrained model could in many cases improve the MSE. (See Appendix \ref{appx:MSE}.)
One must be aware that adding the constraint does not make the mean estimation closer to the true value in every dimension.
The main benefit of applying the sum constraint is to reduce uncertainty in the model and such reduction in uncertainty outweighs the marginal increase in bias in cases when the mean estimation is relatively accurate but the uncertainty is high.
Through simulation, we have also shown the effect of applying such sum constraint to some other non-Gaussian models and the theoretical result carried over quite well.

When sampling from a constrained distribution, both the landscape of the distribution and the underlying constraint may pose challenges to the sampling algorithm.
Our proposed algorithm effectively decouples the distribution and the constraint, where one conducts rejection sampling on the distribution but only requires uniform samples from the constraint.
When the landscape of the constraint is easy to traverse, while the distribution is hard to sample, MCMC-based algorithm might struggle to explore the whole space, but the constrained fusion algorithm would not be hindered.
However, due to its rejection sampling nature, the Constrained Fusion algorithm tends to suffer from high rejection rates when the constraint does not lie close enough to the typical set of the full target distribution.
This may be solved by approximating the rejection stages by multiple stages of sequential Monte Carlo, in analogy to the Bayesian Fusion algorithm \citep{dai2023bayesian} which is a sequential adaptation of the Monte Carlo Fusion algorithm \citep{dai2019monte}. This is left to future work.

The constrained imputation framework is applicable to various other problems in the same field with very similar theoretical setups, for example, data imputation \citep{Peppanen.2016} for handling missing smart meter data or non-intrusive disaggregation \citep{zhao2020non} which deals with disaggregation at appliance level, which all admit a least one sum constraint in the problem.
Since readings might also be aggregated over various houses due to the nature of hierarchical energy systems \citep{wang2020regional}, data privacy, or other reasons \citep{Poursharif.2017}, we may also find sum constraints imposed among households at different resolution levels.
These problems, with more unknowns and constraints, would potentially require more efficient simulation methods which are left as future work.

\newpage
{\noindent \LARGE \bf Appendix \vspace{3em}}
\begin{appendix}

{\bf \noindent
The equations in the Appendix are labeled using $(letter.number)$ format (e.g., (A.1)), equations that are labeled with number only are references to the equations in the main manuscript.}

\renewcommand{\theequation}{\thesection.\arabic{equation}}

\section{Measure on Smooth Manifolds} \label{appx:manifold}
In order to extend the result from (\ref{eq:theorem1}) to the case where $(\bmy^{(1)},\dots,\bmy^{(m)})\in \mcal{H}$ for some constraint set $\mcal{H}\subset \mbb{R}^{md}$, one needs to formalize the notation of measure (or integration) on the constraint set $\mcal{H}$.

\begin{definition}
A (second countable, Hausdorff) topological space $\bm{M}$ is an $n$-dimensional \textbf{smooth manifold} if there exists a family of local coordinates $\mcal{A}=\{(U_i,\varphi_i):i\in I\}$ called \textbf{atlas}, where
\vspace{-1.5em}
\begin{itemize}
    \item $\{U_i\}$ is an \textit{open cover} for $\bm{M}$, i.e., $\cup_{i\in I} U_i = \bm{M}$;
    \item $\varphi_i: U_i\rightarrow \mbb{R}^{n}$ is a \textit{continuous bijection} onto its image $\varphi_i(U_i)\subseteq \mbb{R}^n$ with a \textit{continuous inverse};
    \item whenever $U_i\cap U_j\neq \emptyset$, the transition map $\varphi_j\circ \varphi_i^{-1}:\varphi_i(U_i\cap U_j)\rightarrow \varphi_j(U_i\cap U_j)$ is a \textit{smooth bijeciton} with \textit{smooth inverse}.
\end{itemize}
\vspace{-1.5em}
A pair of $(U_i,\phi_i)\in\mcal{A}$ is called a \textbf{chart}.
If there exists an atlas such that $\forall i,j\in I$, $\det(\dd(\varphi_j\circ\varphi_i^{-1})) >0$, then $\bm{M}$ is an orientable manifold.
\end{definition}
From the definition of manifolds, for any set $A$ and chart $(U_i,\varphi_i)$ such that $A\subset U_i$,  the integral on $A\cap U_i\subset \bm{M}$ may be transformed into an integral on $\varphi_i(A\cap U_i)\subseteq \mbb{R}^{n}$.
Thus we may say a set in $\bm{M}$ is measurable if any charted section of it is measurable in $\mbb{R}^{n}$.
\begin{definition}
Let $(\bm{M},\mcal{A})$ be a manifold. A set $A\subseteq \bm{M}$ is measurable if $\forall p\in A$, there is a chart $(\varphi,U)$ such that $p\in U$ and $\varphi(A\cap U)$ is Lebesgue-measurable.
Let $$\mcal{L}(\bm{M}):=\{A\subset \bm{M}: A \text{ is measurable}\}.$$
\end{definition}
The definition for Lebesgue-measurable sets in $\bm{M}$ is independent of the choice of atlas $\mcal{A}$ (see Remark 1.1, Chapter XII \citep{amann2009analysis}) and the set $\mcal{L}(\bm{M})$ is compatible with the Borel $\sigma$-algebra $\mcal{B}(\bm{M})$, $\mcal{B}(\bm{M})\subseteq \mcal{L}(\bm{M})$ (see Propotision 1.2 of Chapter XII \citep{amann2009analysis}).

\begin{lemma}\label{lemma:volume}
    Let $\vec{\bm{h}}:\mbb{R}^{n+k}\rightarrow\mbb{R}^k, \,\,0<k<md $ be a smooth function such that $\forall \bm{u}\in \vec{\bm{h}}^{-1}(0)$, the derivative $ \mathtt{d}\vec{\bm{h}}_{\bm{u}}:\mbb{R}^{n+k}\,\rightarrow\, \mbb{R}^{k}$ is surjective.
    Then, the set
 \begin{equation}
  \mcal{H}:= \vec{\bm{h}}^{-1}(\bm{0})=\left\{\bm{u}\in\mbb{R}^{n+k}: \vec{\bm{h}}(\bm{u})=\bm{0}\right\}, \nonumber
 \end{equation}
 is a $n$-dimensional manifold. 
 Moreover, there exists a \textbf{canonical} volume form $\Vol_{\mcal{H}}$ defined on $\mcal{H}$ such that the integral of $\Vol_{\mcal{H}}$ over $\mcal{H}$, $\int_\mcal{H} \dd \Vol_{\mcal{H}}$, computes the volume $\mcal{H}$ in the Euclidean space. Thus $\Vol_{\mcal{H}}$ acts as the Lebesgue measure on $\mcal{H}$.
\end{lemma}

\begin{proof}
    By the Regular Value Theorem, see for instance Theorem 9.9 of \cite{tu2011manifolds}, $\mcal{H}=\vec{\bm{h}}^{-1}$ is a submanifold of $\mbb{R}^{n+k}$ with dimension $n+k-k=n$.

    If $\mcal{H}$ is endowed with an indefinite-Riemannian metric $g$, a non-degenerate symmetric bilinear map on the tangent vectors, then $g$ defines a unique volume measure that is independent of the choice of the atlas $\mcal{A}$. (See Section 1 of Chapter XII in \cite{amann2009analysis}.)

    Moreover, since $\mcal{H}$ is embedded in the Euclidean space, the Riemannian metric $g$ induced by the Euclidean inner product defines a unique volume measure $\Vol_\mcal{H}$ such that $\Vol_{\mcal{H}}(\mcal{H})$ is exactly the volume occupied by $\mcal{H}$ inside $\mbb{R}^{n+k}$. \flushright{$\square$}
\end{proof}


Lemma \ref{lemma:volume} gives the sufficient condition for which we may identify canonically a measure on the manifold $\mcal{H}$ and define probability distributions on $\mcal{H}$. 
We also give some examples of constraints that satisfy the condition.
Note that not all constraints defined by smooth functions are satisfied.

\begin{example}[Example 1]
Consider $ \vec{\bm{h}}(x_1,x_2):=x_1^2+x_2^2$. Then $ \mathtt{d}\vec{\bm{h}}_{(x_1,x_2)}=\begin{pmatrix}2x_1 & 2x_2\end{pmatrix}$.
Since $(0,0)\in  {\vec{\bm{h}}}^{-1}(0)$ and $ \mathtt{d}\vec{\bm{h}}_{(0,0)} = \begin{pmatrix}0&0 \end{pmatrix}$ is degenerate.
Thus $ {\vec{\bm{h}}}^{-1}(0)$ is not a manifold.
However, $ {\vec{\bm{h}}}^{-1}(c)$, $c>0$ is a manifold.
\end{example}
\begin{example}[Example 2]
 If $ {\vec{\bm{h}}}:\mbb{R}^{n}\rightarrow\mbb{R}^m$ is linear, $n>m$, $
  \exists \bm{A}\in\mbb{R}^{m\times n},\,\,  \vec{\bm{h}}(\bmx)=\bm{A}\bmx$.
Then $  {\mathtt{d}\vec{\bm{h}}}_{\bmx}=\bm{A}$, thus $ {\vec{\bm{h}}}^{-1}(c)$ is a manifold of dimension $n-m$ if and only if $\bm{A}$ is full rank.
\end{example}

\begin{remark}
Locally, on the set $A\subset U_i$, the integral $\int_A \dd \Vol_\mcal{H}$ can be expressed in local coordinates, i.e., in $\mbb{R}^{n}$,
$$\int_A\dd\Vol_{\mcal{H}} = \int_{\varphi_i(A)} \sqrt{|\det(g)|}\dd x_1\dots\dd x_n$$
where $g$ is the Riemann metric associated with the manifold. 
If there is a global parameterization for $\mcal{H}$, namely a single chart $\varphi:\mcal{H}\rightarrow\mbb{R}^{n}$, then the above evaluation applies to any integrable set of $\mcal{H}$.
Luckily, since $\mcal{H}\subseteq \mbb{R}^{n+k}$ is a Lindelöf space, there is a countable atlas and any subset of $\mcal{H}$ may be partitioned into a countable sequence of disjoint sets $B_i:=(A\backslash \bigcup_{j<i}U_j)\cap U_i$, and define $\Vol_{\mcal{H}}(A) = \sum_i \Vol_{\mcal{H}}(B_i)$.
\end{remark}
Now, since for any $A\in\mcal{L}(\mcal{H})$, $\Vol_{\mcal{H}}(A):=\int_A \dd \Vol_{\mcal{H}}$ can be computed with respect to the volume form, we denote $\lambda_{\mcal{H}}$ the \textbf{Riemann-Lebesgue} volume measure of $\mcal{H}$, where 
$$\lambda_{\mcal{H}}(A):=\Vol_{\mcal{H}}(A), \forall A\in\mcal{L}(\mcal{H}).$$  

\begin{theorem}\label{thm:constrained_rn}
Let $\vec{\bm{h}}:\mbb{R}^{n+k}\rightarrow\mbb{R}^{k}$ be a smooth function such that $\vec{\bm{h}}^{-1}(\bm{0})$ is a $n$-dimensional manifold.
Suppose that $f,g:\mbb{R}^{n+k}\rightarrow \mbb{R}_{>0}$ are two density functions fully supported on $\mbb{R}^{n+k}$, with finite integral on $\mcal{H}$.
Then the following holds:
\begin{enumerate}
    \item we may define naturally the measures $P_f,P_g:\mcal{B}(\mcal{H})\rightarrow [0,1]$ induced by restricting $f,g$ on $\mcal{H}$ such that the Radon-Nikodym derivative with respect to the volume measure $\int_\cdot \dd\Vol_{\mcal{H}}$ is proportional to their corresponding density on the full space;
    \item $P_f \ll P_g$ with $$ \frac{\dd P_f}{\dd P_g}\propto \frac{f}{g}.$$
\end{enumerate}
\end{theorem}

\begin{proof}
    \begin{enumerate}
        \item Since $\mcal{H}\subseteq \mbb{R}^{n+k}$, the density $f$ can be restricted onto $\mcal{H}$ with $f_{|\mcal{H}}(p)=f(p)$, $\forall p\in\mcal{H}\subseteq\mbb{R}^{n+k}$. 
        Since $f_{|\mcal{H}}$ has finite integral on $\mcal{H}$, then let $Z_f:= \int_{\mcal{H}}f \dd\Vol_{\mcal{H}}<\infty$.
        We may define the measure $P_f$ as
        $$P_f(A):= \frac{1}{Z_f}\int_A f \dd\Vol_\mcal{H}, \quad \forall A\in\mcal{B}(\mcal{H}).$$
        Thus $\frac{\dd P_f}{\dd \Vol_{\mcal{H}}}=\frac{f}{Z_f}\propto f$.
        \item Observe that, $\forall A\in\mcal{B}(\mcal{H})$
        \begin{align*}
            P_f(A) = & \int_{A}\frac{f}{Z_f}\dd\Vol_{\mcal{H}}\\
            = & \int_A \frac{Z_g f}{Z_f g} \frac{g}{Z_g} \dd \Vol_{\mcal{H}}.
        \end{align*}
        Clearly, $P_f \ll P_g$ with $\frac{\dd P_f}{\dd P_g}=\frac{Z_g f}{Z_f g}\propto \lambda$. \hfill $\square$
    \end{enumerate}
\end{proof}


\begin{remark}
The integrability condition is automatically obtained when all $f_i$ are continuous and $\mcal{H}$ is compact.
For non-compact $\mcal{H}$, e.g., the linear constraint case is usually not hard to verify, otherwise, one needs to be careful so that the constrained densities are well-defined.
\end{remark}

\begin{remark}
Recall our target is to sample from the following density function (\ref{eq:cons1})
$$
f_{\mcal{H}}\left(\bmy^{(1)},\dots,\bmy^{(m)}\right)\propto f_1(\bmy^{(1)})f_2(\bmy^{(2)})\cdots f_m(\bmy^{(m)}) \mbb{I}_{\bm{y}^{(1:m)}\in\mcal{H}}.
$$
From Lemma \ref{lemma:volume}, we know that when the Riemannian metric $g$ associated with the Riemannian manifold $\mcal{H}$ is fixed, then $(\mcal{H},g)$ admits a unique canonical measure which we use as the dominating measure of the density in (\ref{eq:cons1}).
Since we are sampling the random variables $y^{(1)},\dots,y^{(m)}$, it is natural to consider the Riemannian metric induced by the state space of $(\bmy^{(1)},\dots,\bmy^{(m)})$, namely the Euclidean space.
Now (\ref{eq:cons1}) defines a density with respect to the volume measure of $\mcal{H}$ in the Euclidean space.
\end{remark}

\subsection{Invariance under Reparameterisation}
One important fact to note is that the density part of the formulation of $f_{\mcal{H}}(\bm{y})$ on manifold $\mcal{H}\subset \mbb{R}^n$
$$
    f_{\mcal{H}}(\bm{y}) \propto f(\bm{y})\mbb{I}_{\bm{y}\in\mcal{H}}
$$
is invariant under reparameterisation of $\mcal{H}$, with respect to the canonical volume measure of the manifold ${\cal H}$.

To see this, suppose we have two global parameterisations of $\mcal{H}$, $(x_1,\dots,x_d)$ and $(x'_1,\dots,x'_d)$, with their Riemannian metrics denoted $g$ and $g'$ under their parameterisation respectively.
Since the canonical volume measure of $\mcal{H}$ is fixed through the embedding of $\mcal{H}$ into the Euclidean space, it must hold for parameterizations (of the canonical measure) that 
$$
\sqrt{|\det(g)|}\dd x_1\dots \dd x_d = \sqrt{|\det(g')|}\dd x'_1\dots\dd x'_d
$$
and hence if the likelihood function $f(\bm{y})$ is expressed as $f(\bm{y})=g(x_1,\dots,x_d)=g'(x'_1,\dots,x'_d)$, then we still have
$$
\int_{\mcal{H}} g(x_1,\dots,x_d) \sqrt{|\det(g)|}\dd x_1\dots \dd x_d = \int_{\mcal{H}} g'(x'_1,\dots,x'_d) \sqrt{|\det(g')|}\dd x'_1\dots\dd x'_d = 1.
$$
Thus the density function is invariant, up to a constant, under reparameterisation. 

In the traditional perspective of multivariate calculus or the change of variable formula in probability theory, we may use $\dd x_1\dots \dd x_d$ as the base measure and $\sqrt{|\det(g)|}g(x_1,\dots,x_d)$ as the likelihood, or $\dd x'_1\dots\dd x'_d$ as the base measure and $ \sqrt{|\det(g')|}g'(x'_1,\dots,x'_d)$ as the likelihood which, however, are not written with respect to the canonical volume measure of the manifold ${\cal H}$.

\begin{example}[Polar Transformation]
Now we consider using the simple polar transformation $(x,y)\rightarrow (r,\theta)$ to explain this, 
\begin{align*}
    x \, = r\cos(\theta) \;\;\;\;
    y \, = r\sin(\theta)
\end{align*}
The density $f_{x,y}(x,y)$ is with respect to the Lebesgue measure $\dd x\dd y$. With the Jacobian of the variable transformation, we have the density for $(r,\theta)$ as
$$
    f_{r,\theta}(r,\theta) = f_{x,y}(r\cos(\theta),r\sin(\theta)) \left|\frac{\partial(x,y)}{\partial(r,\theta)}\right| := rg(r,\theta)
$$
where we define $g(r,\theta) = f_{x,y}(r\cos(\theta),r\sin(\theta))$.
However, the additional $r$ term does not belong to the density function, instead, it is included in the canonical volume measure of the constraint ${\cal H}$. In the equality
$$
    f_{x,y}(x,y)\cdot \dd x\dd y = g(r,\theta) \cdot r\dd r\dd\theta
$$
the densities $f_{x,y}(x,y)$ and $g(r,\theta)$ are equal, and both of the measures $\dd x\dd y$ and $r\dd r\dd\theta$ represent the ($\sigma$-finite) uniform measure on the Euclidean of $(x,y)$ space, i.e., the canonical volume measure of the base manifold. So re-parameterisation does not affect the expression of the density. 

Linking this example to our algorithm, the uniform samples from Step 2 of Algorithm \ref{alg:CMCF-2} will be generated from the Euclidean $(x,y)$ space  under the Lebesgue measure $\dd x\dd y = r\dd r\dd\theta$, rather than under the Lebesgue measure of $\dd r\dd\theta$ which instead represents the uniform distribution on the Euclidean $(r,\theta)$ space. 
\end{example}

\section{Constrained Fusion Sampler}
\label{appx:proof1}
\subsection{Rejection Sampling for Diffusions}\label{appx:proof Lemma 2}
\begin{condition}
Let $\bm{\alpha}_i(\bm{u}) = \bmD A_i(\bm{u})$.
\begin{enumerate}[label=(\roman*)]
 \item \[\exp\left\{\int_0^{T}\bm{\alpha}_i(\bm{\omega}^{(i)}_s)\cdot\dd \bm{\omega}^{(i)}_s -\int^{T}_0\frac{1}{2} \left\|\bm{\alpha}_i(\bm{\omega}^{(i)}_s)\right\|^2\dd s \right\}\] is a martingale with respect to $\mathbb W_i$, the Brownian motion measure.
 \item $\bm{\alpha}_i$ is continuously differentiable in all its arguments.
 \item The function
 \begin{equation}
 \phi_i(\bm{u}) = \frac{1}{2}\left[\|\bm{\alpha}_i(\bm{u})\|^2+ \text{\bf div } \bm{\alpha}_i(\bm{u})\right] - l_i \geq 0, \nonumber
\end{equation}
for some $l_i$ and for any $\bm u$, where {\bf div} means the divergence of $\alpha_i$.
\end{enumerate}
\label{condition:regularitycondition}
\end{condition}

\begin{proof}[Proof of Lemma \ref{theorem1}]
 Let $\mbb{P}_i^{(T,\bmx,\bmy)}$ be the law of the diffusion bridge $\bmX^{(i)}$ given the length $T$ and end points $\bmx,\bmy$. Let $\mbb{W}^{(T,\bmx,\bmy)}$ be the law of the Brownian bridge conditioned on length and two ends $(T,\bmx,\bmy)$. From Lemma 1 in \citep{beskos2006exact},
 \begin{equation}
      \frac{\dd \mbb{P}_i^{(T,\bmx,\bmy)}}{\dd\mbb{W}^{(T,\bmx,\bmy)}}= \frac{\mcal{N}_{T}(\bmy-\bmx)}{p_i(\bm{y}^{(i)}| {\bm x}^{(i)})} \times\exp\left\{A_i(\bmy^{(i)})-A_i(\bmx^{(i)})-\int_0^{T}\left( \phi_i(\bmx^{(i)})+l_i\right)\dd s\right\}.\label{eq:lemma1}
 \end{equation}
 Rearrange (\ref{eq:lemma1}) and take expectation with respect to $\mbb{W}^{(T,\bmx,\bmy)}$ leads to equation (\ref{eq:transition}).
 Recall from (\ref{eq:biased_langevin}) and (\ref{eq:wiener}) that without constraint, $g$ and $h$ are defined as
 \begin{equation*}
g\left(\bm{x}^{(1)},\cdots, \bm{x}^{(m)}, \bm{y}^{(1)},\cdots,\bm{y}^{(m)}\right) \propto 
\prod_{i=1}^m f_i^2(\bm{x}^{(i)}) p_i(\bm{y}^{(i)}|\bm{x}^{(i)}) \frac{1}{f_i(\bm{y}^{(i)})},
\end{equation*}
and
\begin{equation*}
h(\bm{x}^{(1)},\cdots\bm{x}^{(m)},\bm{y}^{(1)},\cdots,\bmy^{(m)}) =\prod_{i=1}^m f_i(\bmx^{(i)})(2\pi T)^{-1/2}\exp\left[ -\frac{\|\bmy^{(i)}-\bmx^{(i)}\|}{2T}\right].
\end{equation*}
Substitute (\ref{eq:transition}) into (\ref{eq:biased_langevin}) and we get (\ref{eq:theorem1}).
\hfill$\square$
\end{proof}

\begin{proof}[Proof for Corollary \ref{cor:constrained_rnd}]
    Apply Theorem \ref{thm:constrained_rn} to the result of Lemma \ref{theorem1}. \hfill $\square$
\end{proof}

\subsection{Diffusion Bridge Sampling via Poisson Thinning} \label{appx:poisson_process}
\begin{proposition}
Let $\omega\in C([0,T],\mbb{R})$, $\Phi$ be a Poisson point process of intensity $1$ on the space $[0,T]\times [0,M]$, where $M$ is the upper bound of the function $\phi$.
Let $A:= \{(t,u)\in [0,T]\times [0,M] : u\leq \phi(\omega(t))\}$ denote the region under the curve $\phi(\omega(t))$, then
\begin{equation*}
    \mbb{P}(N(A) = 0 | \omega) = \exp\left\{-\int_0^T \phi(\omega(t))\dd t\right\}
\end{equation*}
\end{proposition}
\begin{proof}
By the definition of $\Phi$, the random variable $N(A) | \omega$ is Poisson with intensity $\int_A \dd \lambda$ with $\lambda$ is the Lebesgue measure on $[0,T]\times [0,M]$.
Thus
$$\mbb{P}(N(A)=0 | \omega) = \exp\left\{ - \int_A \dd\lambda\right\} = \exp\left\{- \int_0^T \phi(\omega(t))\dd t\right\}.$$
\hfill $\square$
\end{proof}
Recall that conditioned on the underlying process $\bm{\omega}_s^{(i)}$, the acceptance probability is given by 
\begin{equation}\tag{\ref{eq:theorem1}}
    \exp\left( - \sum_{i=1}^m \int_0^{T} \phi_i(\bm{\omega}_s^{(i)}) \dd s \right).
\end{equation}
Thus, when the functions $\phi_i$ are all bounded above, we can simulate the event by implementing the proposition:
\begin{enumerate}
    \item For each $i$, find upperbound $\phi_i(u)\leq M^{(i)}, \forall u$;
    \item simulate a Poisson point process $\Phi^{(i)}=\{(t_1,u_1),\dots,(t_\kappa,u_\kappa)\}$ on $[0,T]\times[0,M_i]$;
    \item simulate a Brownian bridge $\bm{\omega}^{(i)}$ connecting $\bmx^{(i)}$ and $\bmy^{(i)}$ on the time points specified by the first coordinate $t_k$ of $\Phi^{(i)}$\;
    \item accept if no point of $\Phi^{(i)}$ lies below $\phi_i(\bm{\omega}_s)^{(i)}$, i.e., $u_k > \phi_i(\bm{\omega}_{t_k}^{(i)})$ for every $k\in\{1,\dots,\kappa\}$.
\end{enumerate}
Note that we only need the value of the proposal Brownian bridge at the time points specified by the Poisson point process.
\begin{figure}
    \centering
    \includegraphics[width=0.7\linewidth]{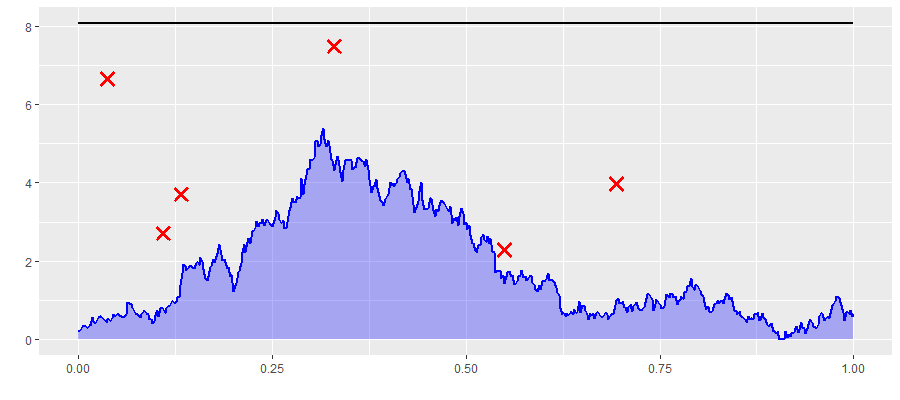}
    \caption{An example of the Poisson process rejection. The blue curve is the trajectory of $\phi(\omega_s)$ and the red crosses are the generated points. Since all the points are above the curve, the acceptance step is passed.}
    \label{fig:appx_poisson_rejection}
\end{figure}

\subsection*{For Unbounded $\phi_i$}
When $\phi_i$ is potentially unbounded, it is impossible to find a finite $M^{(i)}$ that upperbounds $\phi_i(\omega_s^{(i)})$ where $\omega_s^{(i)}$ is the trajectory of a Brownian bridge, unless we simulate $\omega_s^{(i)}$ conditioned on it not leaving a bounded interval $[a,b]$.

For simplicity, let's consider simulating a $1$-dimensional Brownian bridge $X_s$ connecting $x$ at $s=0$ and $y$ at $s=T$.
Let $0=a_0<a_1<a_2<\dots$ be a monotone increasing sequence of positive numbers, then the space of all Brownian bridges $(X_s)$ connecting $x,y$ with time length $T$ may be partitioned into the following sequence of sets, for $i=1,2,\dots$
\begin{align*}
    U_i := &\, \{(X_s): \inf_s X_s > \min(x,y)-a_{i}, \quad \max(x,y)+a_{i-1} \leq \sup_s X_s < \max(x,y)+a_i\} \\
    L_i := &\, \{(X_s): \sup_s X_s < \max(x,y)+a_{i}, \quad \min(x,y)-a_{i-1} \geq \inf_s X_s > \min(x,y)-a_{i}\}
\end{align*}
Let $U_i\cup L_i$ denote the $i$-th layer, the probability of a Brownian bridge not leaving the $i$-th layer can be simulated exactly with the help of following result.

Use $\mtt{p}(T,x,y,K)$ to denote the probability that a one-dimensional Brownian bridge connecting $x$ and $y$ of time $T$ does not leave the interval $[-K,K]$, $K>\max\{x,y\}$.
Then by, for instance, Theorem 3 of \cite{potzelberger2001boundary}, we have
\begin{lemma}
Define for $j\geq 1$,
\begin{align*}
    \bar{\sigma}_j(T,x,y,K) & = \exp\left\{-\frac{2}{T}[2jK-(K+x)][2jK-(K+y)] \right\} \\
    \bar{\tau}_j(T,x,y,K) & = \exp\left\{-\frac{2j}{s}[4jK^2 + 2K(x-y)] \right\}
\end{align*}
and 
\begin{align*}
    \sigma_j(T,x,y,K) &= \bar{\sigma}_j(T,x,y,K) + \bar{\sigma}_j(T,-x,-y,K)\\
    \tau_j(T,x,y,K) & = \bar{\tau}_j(T,x,y,K) + \bar{\tau}_j(T,-x,-y,K).
\end{align*}
Then
\begin{equation}\label{eq:boundary_crossing_1}
    \mtt{p}(T,x,y,K) = 1 - \sum_{j=1}^{\infty}\{\sigma_j(T,x,y,K) - \tau_j(T,x,y,K)\}.
\end{equation}
\end{lemma}
Since $\mtt{p}(T,x,y,K)$ is given by a convergent alternating series, we can simulate $\mtt{p}(T,x,y,K)$ by drawing a uniform $u\sim\mcal{U}[0,1]$ and iteratively compute the upper and lower bounds of $\mtt{p}(T,x,y,K)$ until $u$ crosses one of the bounds.
More importantly, the probability of a Brownian bridge $(X_s)$ connecting $x,y$ with time length $T$ does not leave layer $i$ is given by
$$
\mbb{P}(I\leq i) = \mtt{p}\left(T, \frac{x-y}{2}, \frac{y-x}{2}, \frac{|x-y|}{2}+a_i\right), \quad i\geq 1.
$$ 
where $I$ is the layer of the bridge.
Given $I=i$, the rest of the construction will follow the steps below:
\begin{enumerate}
    \item With probability $\tfrac{1}{2}$, assume $(X_s)\in U_i$ (otherwise assume $(X_s)\in L_i$), which constrains the maximum point (minimum point resp.) of the bridge.
    \item Simulate the maximum (or minimum) and the time it is attained;
    \item Conditioned on the maximum (or minimum) point, a Brownian bridge can be decomposed into \textbf{two} Bessel bridges starting independently at the maximum (or minimum) point and connecting the starting/ending points of the Brownian bridge.
    \item Verify the path between the skeleton points, which are Bessel bridges, indeed does not leave layer $i$ from the other side, i.e., if the bridge is generated conditioned on its maximum, then check if the bridge leaves layer $i$ from below and vice versa.
    \item If the Brownian bridge does not leave layer $i$, check if the Bessel bridges in between do not leave layer $i-1$ from the other side conditioned on it not leaving layer $i$.
    \item If at least one Bessel bridge in between leaves layer $i-1$ from the other side, it means the generated Brownian bridge belongs to $U_i\cap L_i$ and we should reject the trajectory with probability $\tfrac{1}{2}$.
\end{enumerate}

\begin{example}[Example Instance of Step 3]
Suppose one needs to simulate a Brownian bridge $(X_s)_{s\in[0,T]}$ connecting $x$ and $y$ conditioned on hitting its maximum point $M$ at time $0<\tau<T$.
Then generating the trajectory from time $0$ to $\tau$ is equivalent to generating a Bessel bridge $B_s$ connecting $0$ at time $0$ and $M-x$ at time $\tau$, and computing the path $B_s+x$.
Similarly, the trajectory from $\tau$ to $T$ is equivalent to generating another Bessel bridge $B'_s$ connecting $0$ at time $0$ and $M-y$ at time $T-\tau$, and computing $B'_{T-\tau-s}+y$.
The resulting Brownian bridge is pieced together as follows:
$$
    X_s = \begin{cases}
        B_s+x, & s\in[0,\tau];\\
        B'_{T-s}, & s\in(\tau,T].
    \end{cases}
$$
The same (but opposite) procedure follows if the Brownian bridge is conditioned on its minimum.
\end{example}

Let $\mtt{q}(T,x,y,K)$ denote the probability of a Bessel bridge $(B_s)_{s\in[0,T]}$ connecting $x\geq0$ and $y\geq0$ of time length $T$ does not leave interval $(0,K)$ and $\mtt{q}(T,x,y,K;L)$, $K<L$, denote the probability of $(B_s)_{s\in[0,T]}$ does not leave $(0,K)$ conditioned on it not leaving $(0,L)$.
\begin{lemma}
\begin{align*}
    \mtt{q}(T,x,y,K;L) &\, = \frac{y - \sum_{j=1}^\infty \left\{\zeta_j(T,y,K) - \xi_j(T,y,K) \right\}}{y - \sum_{j=1}^\infty \left\{\zeta_j(T,y,L) - \xi_j(T,y,L) \right\}} \\
    \mtt{q}(T,x,y,K) &\, = 1-\frac{1}{y}\sum_{j=1}^\infty \left\{\zeta_j(T,y,K) - \xi_j(T,y,K) \right\} 
\end{align*}
where
$$
\zeta_j(T,y,K) = (2jK-y)\exp\left\{-\frac{2}{T}jK(jK-y)\right\}, \quad \xi_j(T,y,K) = \zeta_j(T,-y,K)
$$
\end{lemma}
The above lemma can be derived by simply noting that a Brownian bridge connecting $x,y\geq0$ conditioned on not leaving interval $(0,K)$ is a Bessel bridge and thus
$$
\mtt{q}(T,x,y,K;L) = \frac{\mtt{p}(T, x-K/2, y-K/2,K/2)}{\mtt{p}(T, x-L/2, y-L/2, L/2)}.
$$
To avoid fully reiterating the result of \cite{beskos2008factorisation}, we have omitted the proofs and most of the technical details, please address the original paper, or see Section 2.3.3 of \cite{hu2023drawing} for a summary.

\subsection{Simulation from Linearly Constrained Proposal}
\label{appx:linear_constraint}

Without loss of generality, suppose that the manifold $\mcal{H}$ is given by the equation $\bm{A}y^{(1:m)}=\bmc$. 
Then 
\begin{equation}
 h \left(\bm{x}^{(1:m)},\bm{y}^{(1:m)}\right) \propto \left(\prod_{i=1}^m f_i(\bmx^{(i)})\right)\, f_{\bmy|\bmx}\left(\bmy^{(1:m)}|\bmx^{(1:m)}\right)\indi_{\bm{c}}(\bm{A}\bmy^{(1:m)}), 
\end{equation}
where
\begin{equation} \label{eq:split}
    f_{\bmy|\bmx}\left(\bmy^{(1:m)}|\bmx^{(1:m)}\right) = \prod_{i=1}^{m}(2\pi T)^{-1/2}\exp\left[ -\frac{\|\bmy^{(i)}-\bmx^{(i)}\|^2}{2T}\right].
\end{equation}
Note that $f_{\bmy|\bmx}$ is a Gaussian distribution, thus $\bmy^{(1:m)}|\bmx^{(1:m)}, \bm{A}\bmy^{(1:m)}=\bmc$ is also Gaussian.
Note that the conditional Gaussian density is given by 
$$
p(Y^{(1)},\dots,Y^{(m)}|\bmx^{(1:m)}, \bm{A}\bmy^{(1:m)}=\bm{c}) = \frac{f_{\bmy|\bmx}\left(\bmy^{(1:m)}|\bmx^{(1:m)}\right) \indi_{\bm{c}}(\bm{A}\bmy^{(1:m)})}{Z_{\mcal{H}}(\bmx^{(1:m)})},
$$
where the normalizing constant $Z_{\mcal{H}}(\bmx^{(1:m)})$ is the density function of random variable $\bm{A}\bmy\sim \mcal{N}\left(\bm{A}\bmx^{(1:m)}, T\bm{A}\bm{A}^{\top}\right)$ evaluated at $\bmc$, i.e.,
\begin{equation}
Z_{\mcal{H}}(\bmx^{(1:m)}) \propto \exp\left[-\frac{1}{2T}(\bm{c}-\bm{A}\bmx^{(1:m)})^\top (\bm{A}\bm{A}^\top)^{-1}(\bm{c}-\bm{A}\bmx^{(1:m)})\right]<1,\label{eq:ap1_eva}
\end{equation}
where $T$ is the length of the diffusion bridges.
To sample from the proposal distribution $h_{\mcal{H}}$, we just need to simulate $\boldsymbol x^{(i)}$ from each $f_i$ and then simulate the Gaussian random variables $\boldsymbol y^{(1:m)}$ given $\boldsymbol x^{(1:m)}$ and $\bm{A}\bmy^{(1:m)}=\bmc$.
Then followed by a rejection step with acceptance probability $Z_{\mcal{H}}\left(\bmx^{(1:m)}\right) < 1$. 
Simulation from linearly constrained Gaussian distribution has been studied in multiple papers \citep{vrins2018sampling,cong2017fast}.
Appendix \ref{appx:gaussian_sample} presents the routine that is utilized in our algorithm.
We also include a toy example in Appendix \ref{appx:toy_ex}.

\subsubsection{Linearly Constrained Gaussian}\label{appx:gaussian_sample}
\noindent
Consider the following constrained Gaussian distribution
\begin{equation}
    \bm{X}\sim\mathcal{N}(\bm{\mu},\bm{\Sigma}) \quad \text{ subject to } \quad \bm{A}X=\bm{c},\,\bm{A}\in\mathbb{R}^{k\times n},\, k<n.\label{eq:simX}
\end{equation}
Since the covariance matrix is always positive definite, $\bm{\Sigma}$ can be decomposed into $\bm{\Sigma} = \bm{UDU}^\top$, where $\bm{D}$ is a diagonal matrix with positive diagonal entries and $\bm{U}$ is orthogonal.
Then
\begin{equation*}
 \bm{\Sigma} = (\bm{UD}^{\frac{1}{2}})(\bm{UD}^{\frac{1}{2}})^\top
\end{equation*}
where $\bm{D}^{\frac{1}{2}}$ is computed by taking square root of each diagonal entry of $\bm{D}$.
Let $\bm{Z}$ follow a standard multivariate Gaussian distribution, then $\bm{X}\stackrel{d}{=}(\bm{UD}^{\frac{1}{2}})\bm{Z}+\bm{\mu}$.
Thus the simulation problem (\ref{eq:simX}) is equivalent to simulate
\begin{equation*}
 \bm{Z}\sim\mathcal{N}(0,\bm{I}_d) \quad \text{ given that }\quad (\bm{AUD}^{\frac{1}{2}})\bm{Z}=\bm{c}-\bm{A}\bm{\mu}
\end{equation*}
To simplify the notation, let $\bm{B} = \bm{AUD}^{\frac{1}{2}}$, $\bm{\alpha} = \bm{c}-\bm{A}\bm{\mu}$ and instead consider the following problem:
\begin{equation}\nonumber
 \bm{Z}\sim\mathcal{N}(0,\bm{I}_n) \quad \text{ given that }\quad \bm{B}\bm{Z}=\bm{\alpha}
\end{equation}
Let $\bm{B}=\bm{PWQ}^\top$ be a singular value decomposition of $\bm{B}$, where $\bm{P}\in\mbb{R}^{k\times k}$, $\bm{Q}\in\mbb{R}^{n\times n}$ are orthogonal matrices and $\bm{W}\in\mbb{R}^{k\times n}$ is a rectangular diagonal matrix with non-negative entries on the diagonal, i.e.
\begin{equation*}
 \bm{W}=\begin{bmatrix}
    w_1 & 0 & \cdots & 0 &0 & \cdots &0 \\
    0  & w_2 & \ddots & \vdots & 0 & \cdots & 0\\
    \vdots & \ddots &  \ddots & 0&0 & \cdots & 0 \\
    0 & \cdots & 0 & w_{k} & 0 & \cdots & 0 \\
   \end{bmatrix}, w_i \geq 0,\, i\in{1,2,\dots,k}
\end{equation*}
Then the constraint can be expressed as
\begin{equation}\label{eq:constraint}
 \bm{W}(\bm{Q}^\top \bm{Z}) = \bm{P}^\top \bm{\alpha}
\end{equation}
Let $\bm{Y}:=\bm{Q}^\top \bm{Z}\sim\mathcal{N}(\bm{0},\underbrace{\bm{Q}^\top \bm{I}_n\bm{Q}}_{=\bm{I}_n})$. 
Let $y_i$ denote the $i$th element of $\bm{Y}$, and $\tilde{\alpha}_i$ denote the $i$th element of $\bm{P}^\top\bm{\alpha}$. 
Then the constraint (\ref{eq:constraint}) is deterministic
\begin{align*}
 y_1 & = \tilde{\alpha}_1/w_1 \\
 \vdots & \\
 y_k & = \tilde{\alpha}_m/w_k
\end{align*}
Since $\bm{Y}$ is a standard multivariate Normal distribution, thus $[Y_{k+1},\dots,Y_n]$ conditioned on $[Y_1,\dots,Y_k]$ is still a standard Normal distribution. 
Thus the simulation can be done as follows:
\begin{enumerate}[label=(\roman*)]
 \item Compute the deterministic terms of $\bm{Y}$, $y_1,\dots,y_k$;
 \item Simulate the rest of $\bm{Y}$ given the deterministic terms;
 \item Recover $\bm{Z} = \bm{Q}\bm{Y}$;
 \item Recover $\bm{X}=(\bm{UD}^{\frac{1}{2}})\bm{Z}+\bm{\mu}$
\end{enumerate}
\begin{remark}
 In the constrained fusion algorithm, the covariance matrix $\bm{\Sigma}$ is always a diagonal matrix thus decomposition of $\bm{\Sigma}$ is not required.
\end{remark}

\subsection{Simulation from Spherically Constrained Proposal}\label{appx:spherical_constraint}
Let $\mcal{S}_{\bmc,r}^{p-1} := \{ \bmx\in\mbb{R}^p : \|\bmx-\bmc\|_2=r\}$ denote the $(p-1)-$sphere centred at $\bmc$ with radius $r>0$.
If we replace the linear constraint with a spherical constraint, the Brownian motion proposal is now given by
\begin{equation} \label{eq:spherical_1}
     h \left(\bm{x}^{(1:m)},\bm{y}^{(1:m)}\right) \propto \left(\prod_{i=1}^m f_i(\bmx^{(i)})\right)\, \prod_{i=1}^{m}(2\pi T)^{-1/2}\exp\left[ -\frac{\|\bmy^{(i)}-\bmx^{(i)}\|^2}{2T}\right] \delta_{ \mcal{S}_{\bmc,r}^{md-1}}\left(\bmy^{(1:m)}\right)
\end{equation}
where $\bmx^{(1:m)},\bmy^{(1:m)}\in\mbb{R}^{md}$ being the start and end points of the $m$ separate Brownian bridges. 
Notice that the spherical constraint on $\bmy^{(1:m)}$ can be rearranged into an equation of form $\sum_i (y^{(i)})^2 = \sum_i a_iy^{(i)} + b$. 
Thus all the second-order terms in the exponential in (\ref{eq:spherical_1}) can be removed and the resulting density function resembles the density of a von Mises-Fisher distribution.
\begin{definition}[Scaled and Shifted von Mises-Fisher Distribution]
Let $I_{\nu}$ denote the modified Bessel function of the first kind at order $\nu$.
The von Mises-Fisher distribution on $(d-1)-$sphere, denoted by $\vmf(\bm{\mu}, \kappa, \bm{c}, r)$, where $\bm{c}\in \mbb{R}^d$, $\bm{\mu}\in\mcal{S}_{\bmc,r}^{d-1}$, $\kappa\geq 0$, $r>0$, has density function
\begin{equation}
    f_{\vmf}(\bm{x})= \frac{C_d(\kappa)}{r} \exp\left(\frac{\kappa}{r}(\bm{\mu}-\bmc)^{\top}(\bmx-\bmc)\right), \quad \bmx\in\mcal{S}_{\bmc,r}^{d-1}
\end{equation}
where
$$
C_d(\kappa) = \frac{\kappa^{d/2-2}}{(2\pi)^{d/2}I_{d/2-1}(\kappa)},
$$
\end{definition}
To sample from $h_{\mcal{H}}$, we may consider the function $\tilde{h}$ given by
\begin{equation}
    \tilde{h}(\bmx^{(1:m)},\bmy^{(1:m)}) = \prod_{i=1}^m f_i(x^{(i)}) f_{\vmf}\left(\bmy^{(i)}; \bm{\mu}(\bmx^{(1:m)}), \kappa, \bm{c}, r\right) 
\end{equation}
which is also defined on the product space of $\mbb{R}^{md}\times \mcal{S}_{\bmc,r}^{{md}-1}$.
\begin{lemma}\label{lemma:vmf_distribution}
Let $\alpha:=\|\bmx^{(1:m)}-\bmc\|_2^{-1}$, $\bm{\mu}(\bmx^{(1:m)}) = \bmc + \alpha(\bmx^{(1:m)}-\bmc)$, $\kappa=\frac{r^2}{\alpha T}$, we have
\begin{equation}
    \frac{h_{\mcal{H}}(\bmx^{(1:m)},\bmy^{(1:m)})}{\tilde{h}(\bmx^{(1:m)},\bmy^{(1:m)})} \propto \exp\left(--\frac{1}{2T}\|\bmx^{(1:m)}-\bmc\|^2\right)
\end{equation}
for any $(\bmx^{(1:m)},\bmy^{(1:m)})\in \mbb{R}^{md}\times \mcal{S}_{\bmc,r}^{md-1}$.
\end{lemma}
Thus, by Lemma \ref{lemma:vmf_distribution}, we can use the von Mises-Fisher proposal to generate samples that land on a $(m-1)-$sphere followed by a rejection step with acceptance probability given by $Z_{\mcal{H}}(\bmx^{(1:m)})\propto \exp\left(--\frac{1}{2T}\|\bmx^{(1:m)}-\bmc\|^2\right)$.

\begin{proof}
    Recall that
    $$
    h(\bmx^{(1:m)},\bmy^{(1:m)}) \propto \prod_{i=1}^m f_i(\bmx^{(i)}) \exp\left(-\frac{1}{2T}\|\bmx^{(i)}-\bmy^{(i)}\|^2\right)
    $$
    defined on $\mbb{R}^{md}\times \mcal{S}_{c,r}^{p-1}$, i.e., $\|\bmy^{(1:m)}-\bmc\|=r$.
    Note that on the constraint, 
    $$
     \|\bmy^{(1:m)}\|^2 = r^2+2\bmy^{\top}\bmc -\|\bmc\|^2.
    $$
    Thus we may remove all the second-order terms in the proposal $h$
    \begin{align*}
     \exp\left[-\frac{1}{2T}\|\bmy^{(1:m)}-\bmx^{(1:m)}\|^2\right] &= \exp\left[-\frac{1}{2T}\|\bmy^{(1:m)}\|^2 - 2(\bmy^{(1:m)})^{\top}\bmx^{(1:m)} + + \|\bmx^{(1:m)}\|^2 \right] \\
     & =  \exp\left[-\frac{1}{2T} (r^2 + 2(\bmy^{(1:m)})^{\top}(\bmc-\bmx^{(1:m)}) + \|\bmx^{(1:m)}\|^2 -\|\bmc\|^2) \right] \\
     & = \exp\left[-\frac{1}{2T} (r^2 + 2(\bmy^{(1:m)} - \bmx^{(1:m)}-\bmc)^{\top}(\bmc-\bmx^{(1:m)}) \right] \\
     & \propto \exp\left[-\frac{1}{2T} (2(\bmy^{(1:m)} - \bmx^{(1:m)}-\bmc)^{\top}(\bmc-\bmx^{(1:m)}) \right].
    \end{align*}
    Now
    \begin{align*}
       \frac{h(\bmx^{(1:m)},\bmy^{(1:m)})}{\tilde{h}(\bmx^{(1:m)},\bmy^{(1:m)})} \propto & \exp\left[ -\frac{1}{2T} (2(\bmy^{(1:m)} - \bmx^{(1:m)}-\bmc)^{\top}(\bmc-\bmx^{(1:m)}) - \frac{\kappa}{r^2}(\mu(\bmx^{(1:m)}) - \bmc)^{\top}(\bmy^{(1:m)}-\bmc)\right] \\
       = & \exp\left[ -\frac{2}{2T} ((\bmy^{(1:m)} -\bmc)^{\top}(\bmc-\bmx^{(1:m)}) - \frac{1}{2T}\|\bmx^{(1:m)}-\bmc\|^2 - \frac{\kappa}{r^2}(\mu(\bmx^{(1:m)}) - \bmc)^{\top}(\bmy^{(1:m)}-\bmc)\right] \\
       = & \exp\left(-\frac{1}{2T}\|\bmx^{(1:m)}-\bmc\|^2\right)\exp\left[(\bmy^{(1:m)}-\bmc)^{\top}\left(\frac{1}{T}(\bmx^{(1:m)}-\bmc)-\frac{\kappa}{r^2}(\mu(\bmx^{(1:m)}) - \bmc)\right)\right].
    \end{align*}
    By letting 
    $$
    \mu(\bmx^{(1:m)}) = \bmc + \alpha(\bmx^{(1:m)}-\bmc),\,\,\,\, \alpha=\|\bmx^{(1:m)}-\bmc\|^{-1} \text{ and } \kappa = \frac{r^2}{\alpha T},
    $$
    the second exponential term will be canceled, giving the desired expression.
\end{proof}

\subsection{Simulation from Arbitrary Manifold}
For arbitrary manifold constraints, we rely on other MCMC algorithms to first generate samples uniformly from the constraint, for instance, the Constrained HMC algorithm\citep{lelievre2019hybrid}.
The benefit of this compared with direct sampling from the constrained target is that one only needs to verify the convergence property of the MCMC sampler on the constrained uniform case.
Thus for any target distribution, only one sampler needs to be tuned for each type of constraint up to translation, rotation, scaling, etc.
We summarize the algorithm in Algorithm \ref{alg:CMCF-2}.
\begin{algorithm}[t]
 \SetAlgoLined
 \SetKwInOut{Input}{input}
\Input{Manifold Constraint $\mcal{H}$; component distributions $f_i,i=1,\dots,C$; parameter $T$}
 Simulate, for each $1\leq i\leq m$, $\bm{x}^{(i)}\sim f_i(\cdot)$ \;
 Simulate $\bm{y}=(\bmy^{(1)},\dots,\bmy^{(m)})\sim \mcal{U}(\mcal{H})$, the uniform distribution on constraint set $\mcal{H}$\;
 Simulate a uniform random variable $U_1\in\mcal{U}[0,1]$\;
 \uIf{$\log U_1\leq \|\bmy^{(1:m)}-\bmx^{(1:m)}\|_2^2/2T$}{
    \For{$i=1,...,m$}{
        Simulate a Brownian Bridge of length $T$ connecting $\bm{x}^{(i)}$ and $\bm{y}^{(i)}$\;
    }
    Let $U_2\in\mcal{U}[0,1]$ and simulate the event ${\cal I}$ given by expression (\ref{eq:theorem1}), see Appendix \ref{appx:poisson_process}\;
    \uIf{$\mcal{I}$ is true}
    {
        Accept and return $\bm{y}^{(1:m)}$\;
    }
    \Else{
        Go back to step 1\;
    }
 }
 \Else {
    Go back to step 1\;
}
\caption{Constrained Fusion Sampler for Case 2}
\label{alg:CMCF-2}
\end{algorithm}

\section{Toy example for constrained sampling}\label{appx:toy_ex}
\begin{figure}[tb]
 \centering
 \begin{subfigure}[t]{0.48\linewidth}
  \centering
  \includegraphics[width=\linewidth]{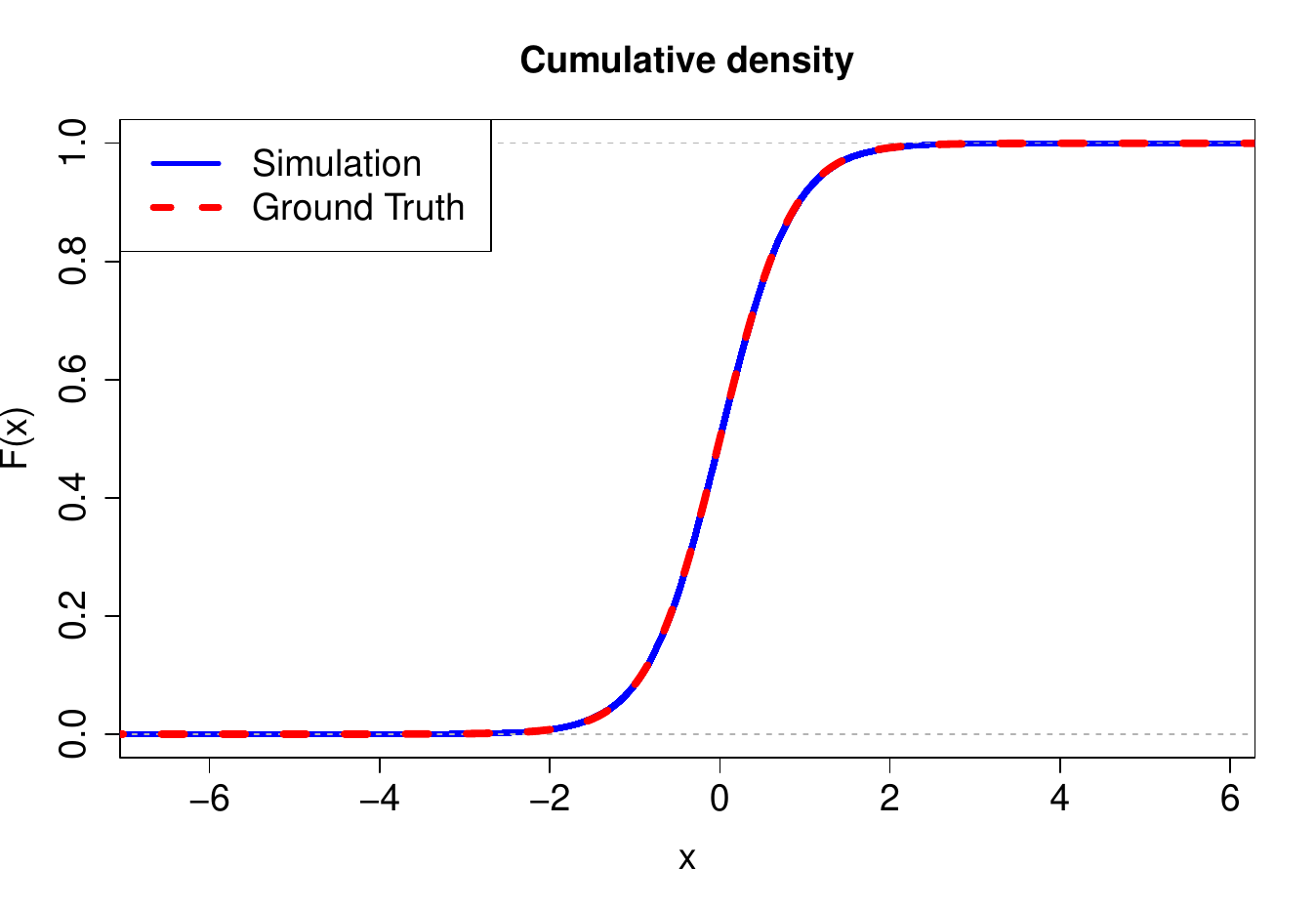}
  \caption{}
  \label{fig:cdf_1}
 \end{subfigure}
 \hfill
  \begin{subfigure}[t]{0.48\linewidth}
  \centering
  \includegraphics[width=\linewidth]{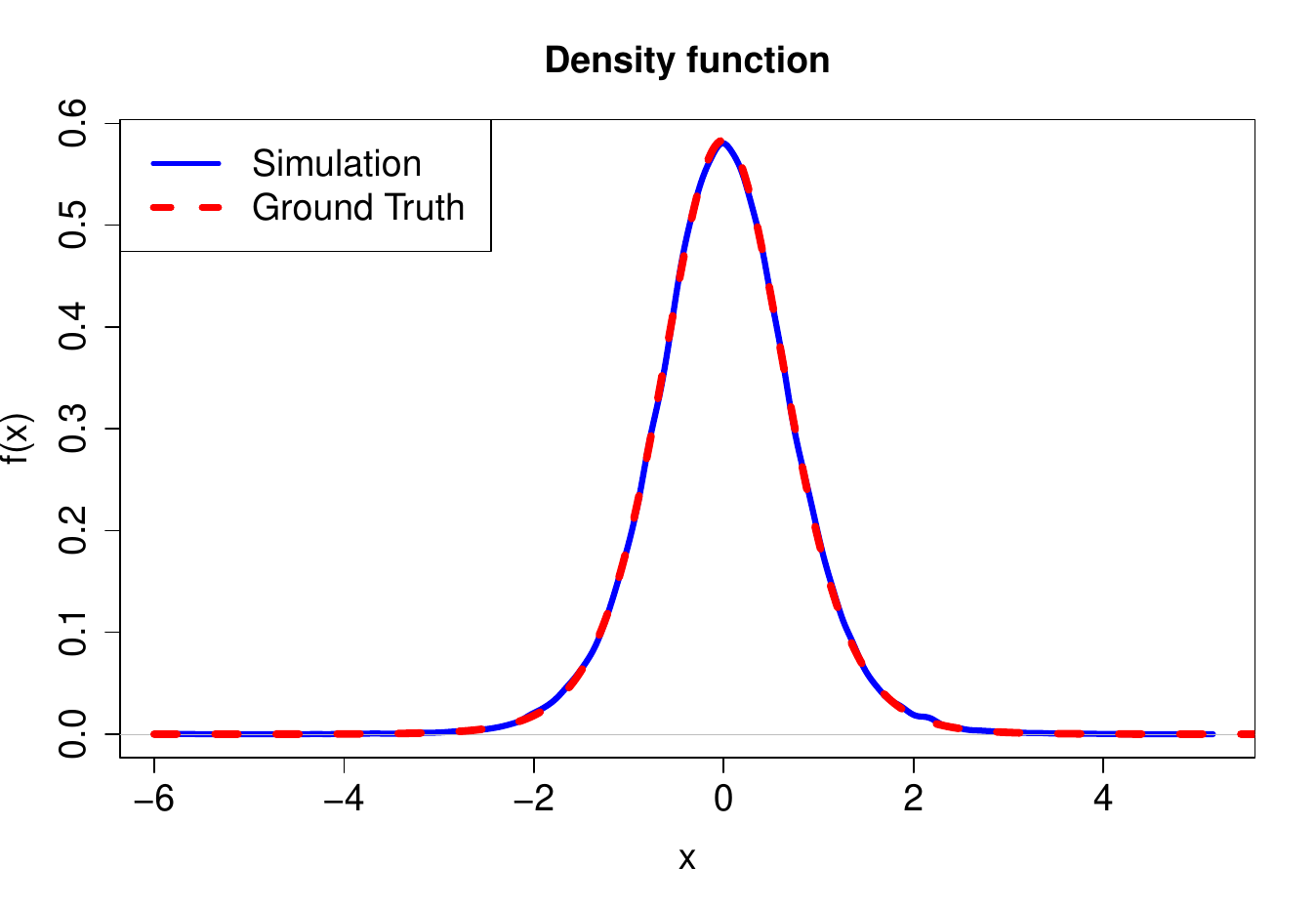}
  \caption{}
  \label{fig:pdf_1}
 \end{subfigure}
 \caption{CDF and PDF of simulated data against ground truth}
\end{figure}
We test the algorithm on a relatively simple setting.
Consider the density function
\begin{equation}
 f(x_1,x_2)\propto \left(1+\frac{x_1^2}{3}\right)^{-2} \left(1+\frac{x_2^2}{5}\right)^{-3}\indi_{x_1+x_2=0}\label{eq:toy_original}
\end{equation}
which is a product density of two distributions Student-${\mathtt T}_3(0)$ and Student-${\mathtt T}_5(0)$ subject to the constraint $x_1+x_2=0$.
We may resolve the constraint and rewrite the density as
\begin{equation}
 f(x_1)\propto \left(1+\frac{x_1^2}{3}\right)^{-2} \left(1+\frac{x_1^2}{5}\right)^{-3}.\label{eq:toy_simp}
\end{equation}
The normalizing constant can be computed numerically.

In the simulation, Algorithm \ref{alg:CMCF-1} is applied to sample from (\ref{eq:toy_original}), gathering 10,000 samples. 
Using these samples, we compute the fitted density function and cumulative function using \textit{ecdf} and \textit{density} function provided by R.
The fitted functions are plotted against the ground truth in Fig. \ref{fig:cdf_1} and \ref{fig:pdf_1}.
From the figures we can see that the fitted densities exactly match the ground truth, hence the validity of the algorithm is verified not only by theory but also by simulation.

\section{Generalized Logistic Distribution}\label{appx:genlog}

\subsection{Basics of Generalized Logistic Distribution}
\begin{definition}[Generalized Logistic Distribution]
    Let $\alpha,\beta,\gamma>0$, $C\in\mbb{R}$, and $Y_1\sim\Gamma(\alpha,1)$, $Y_2\sim\Gamma(\beta,1)$.
    Let 
    $$X:=\gamma \log\left(\frac{X_1}{X_2}\right) + C,$$
    then $X$ is said to follow a Generalized Logistic distribution with parameter $(\alpha,\beta,\gamma,C)$, denoted $X\sim \GenLog(\alpha,\beta,\gamma,C)$.
\end{definition}
\begin{proposition}
    Let $X\sim \GenLog(\alpha,\beta,\gamma,C)$, then the density function $f_X(x)$ is given by
    $$
    f_X(x) = \frac{\Gamma(\alpha+\beta)}{\gamma\Gamma(\alpha)\Gamma(\beta)}\left[1+\exp\left(-\frac{x-C}{\gamma}\right)\right]^{-\alpha} \left[ 1+\exp\left(\frac{x-C}{\gamma}\right)\right]^{-\beta}.
    $$
    Moreover, the first four cumulants (or centralized moments) are given by
    \begin{align*}
        \kappa_1 = & \mbb{E}[X] = C+\gamma(\Psi(\alpha)-\Psi(\beta))\\
        \kappa_2 = & \Var(X) = \gamma^2(\Psi'(\alpha) + \Psi'(\beta)) \,> \,0\\
        \kappa_3 = & \Skew(X) = \gamma^3(\Psi''(\alpha)-\Psi''(\beta))\\
        \kappa_4 = & \kurt(X) = \gamma^4(\Psi'''(\alpha) - \Psi'''(\beta)) \, > \, 0
    \end{align*}
    where $\Psi(u) = \frac{\mtt{d}}{\mtt{d}u}\log\Gamma(u)$ is the "digamma" function.
\end{proposition}
\begin{proof}
See, for instance, Section 4 of \cite{halliwell2018log}.
\end{proof}

The generalized logistic distribution can fit both positively and negatively skewed data by properly setting the values of $\alpha$ and $\beta$ (or unskewed data with $\alpha=\beta$).
The fourth cumulant is always positive, meaning that it has a heavier tail than a Gaussian distribution which is desirable for modeling non-Gaussian residues.
It is also desirable that the $\phi$ function (used in the second rejection step of Algorithm \ref{alg:CMCF-1}) has a global bound of $\max\{\alpha^2,\beta^2\}/(2\gamma^2)$.
This means that the algorithm will be more consistent and efficient when generating samples from this distribution, compared with other distributions that have unbounded $\phi$.

\subsection{Generalized Logistic Distribution Regression}
Here, we will only present a naive way to fit regression and distribution parameters for an auto-regressive model (or a linear regression model) with an identity link function where the residue error is modeled by the \emph{Generalized Logistic Distribution}.

Suppose that we have a time series data $Y_t, t=\in\{1,\dots,n\}$ and the AR model is order $k$.
Let $\Xi_t\in\mbb{R}^d$ denote the extra regressors for the prediction of time $t$.
Let $\Phi\in\mbb{R}^{k}$ be the AR coefficients and $\psi\in\mbb{R}^{d}$ be the regression coefficients
Then the regression model can be written as
\begin{equation} \label{eq:appx_regression}
    Y_t = \underbrace{\psi_0 + \sum_{r=1}^k \Phi_r Y_{t-r} + \Xi_t \bm{\psi}}_{\mu_t} + \epsilon_t, \qquad \epsilon_t\sim \GenLog(\alpha,\beta,\gamma,C).
\end{equation}
Thus the goal is to fit the coefficients $\Phi$ and $\psi$, and the distribution parameters $\alpha,\beta,\gamma,C$.

The naive although not optimal way is to first determine the parameters $\Phi$ and $\psi$ by treating it as a simple linear regression.
By vectorizing and stacking (\ref{eq:appx_regression}), we can rewrite the equation into
$$
\bmY = \bmX \tilde{\Phi} + \bm{\epsilon}
$$
where $\bmY$ is the response vector $\left[Y_{k+1},\dots,Y_n\right]^{\top}\in\mbb{R}^{n-k}$, $\bmX\in\mbb{R}^{(n-k)\times (1+d+k)}$ is the design matrix with an added column of ones for intercept and $\tilde{\Phi} = [\psi_0, \Phi,\psi]^{\top}\in\mbb{R}^{1+d+k}$.
Use $(\bmX^{\top}\bmX)^{-1}\bmX^{\top}\bmy$ to approximate $\tilde{\Phi}$ and compute the residues $\bm{\epsilon}$.

Finally, we can determine the distribution parameters $\alpha,\beta,\gamma,C$ according to the residues.
Let $\hat{\kappa}_2$, $\hat{\kappa}_3$ and $\hat{\kappa}_4$ be the variance, skewness and excess kurtosis of the residue vector $\bm{\epsilon}$, we can implement a non-linear solver to solve the system of equations:
\begin{equation*}
    \begin{cases}
        \gamma^2(\Psi'(\alpha) + \Psi'(\beta)) = \hat{\kappa}_2\\
        \gamma^3(\Psi''(\alpha)-\Psi''(\beta)) = \hat{\kappa}_3\\
        \gamma^4(\Psi'''(\alpha) - \Psi'''(\beta)) = \hat{\kappa}_4. 
    \end{cases}
\end{equation*}
After solving for $\alpha$, $\beta$, $\gamma$, set
$
C = -\gamma(\Psi(\alpha)-\Psi(\beta))
$
to make the residue distribution have zero mean.
Since shifting the Generalized logistic distribution is equivalent to shifting its parameter $C$, the resulting predictor $Y_t=\mu_t +\epsilon_t \sim \GenLog(\alpha,\beta,\gamma, C+\mu_t)$.

\begin{remark}
    In practice, the excess kurtosis from the residues may be negative and the solver will not produce a valid solution.
    We manually set $\hat{\kappa}_4$ to be $0$ in these cases and use a non-linear optimizer to produce an approximate fit as close as possible.
\end{remark}
\begin{remark}
    In order for the sampling algorithm to have stable performance, we need to avoid the situation where $\alpha$, $\beta$, and $\gamma$ are fitted to insensible values.
    To make the fitting more robust, we utilized an optimizer to minimize the quadratic difference between the fitted cumulants $\kappa_{2:4}$ and the empirical cumulants $\hat{\kappa}_{2:4}$ subject to an $L_2$ regularization to keep the parameters small.
    Thus, the problem becomes
    \begin{align*}
        \argmin_{\alpha,\beta,\gamma} \quad & \left[\gamma^2(\Psi'(\alpha) + \Psi'(\beta)) - \hat{\kappa}_2\right]^2 + 
          \left[\gamma^3(\Psi''(\alpha)-\Psi''(\beta)) - \hat{\kappa}_3\right]^2  \\
        & \qquad+\left[\gamma^4(\Psi'''(\alpha) - \Psi'''(\beta)) - \hat{\kappa}_4\right]^2 + \lambda_1(\alpha^2+\beta^2) + \lambda_2\gamma^2 \\
        \text{subject to} \quad & \quad \alpha,\beta,\gamma>0.
    \end{align*}
    $\lambda_1$ and $\lambda_2$ can be very small when the data (residue) is suitable.
    This may require one to apply some scaling to the raw data.
    In our case, $\lambda_1=10^{-3}$ and $\lambda_2=10^{-6}$.
\end{remark}

\section{Simulation Studies}\label{appx:application}
\subsection{Covariates in Study 1} \label{appx:covariates}
The following survey results are used in the Study 1 (Sec \ref{sec:study1}):
\begin{enumerate}
    \item Number of people over 15 years of age in your home;
    \item Number of people under 15 years of age in your home;
    \item Number of bedrooms in your home;
    \item Equipped with a washing machine?
    \item Equipped with tumble dryer?
    \item Equipped with dishwasher?
    \item Equipped with an electric cooker?
    \item Equipped with electric heater (plug-in convector heaters)?
    \item Equipped with stand-alone freezer;
\end{enumerate}

\subsection{Study 2: Max-Min Prediction}\label{sec:study2}
In this example, we consider the energy consumption modelling problem from the Western Power Distribution challenge\footnote{https://codalab.lisn.upsaclay.fr/competitions/213}.
Spikes in energy demand could strain the network and one might mitigate the effect by monitoring the usage and reacting to the surge.
However, monitoring the power usage with high-resolution reading can be expensive since this requires the installation of additional facilities and an ever-expanding data storage system. 
Thus, instead of monitoring with high-frequency in the long term, one might instead want to gather enough data to train a model to impute the high-frequency data and only maintain a low-frequency monitoring system.
The goal is to predict the peaks and troughs of high-frequency time series for each half-hour using its average power consumption.
{The peaks and troughs are measured with respect to the discretized reading at the higher frequency.}

\subsubsection*{Parameter Estimation}
For monitoring peaks and troughs of energy usage, we consider the time series at a much higher frequency than once per day.
The low-frequency observation stream will have a reading every 30 minutes and the high-frequency stream will have a reading every 6 minutes. 
Both time series record the average power usage within the time interval and the goal is to estimate the peak and trough values every 30 minutes in the 6-minute time series.

Using a similar notation as in (\ref{eq:ARmodel}), let $S_t^{(i)}$ denote the 30-minute readings and $Y_t^{(i)}$ denote the 6-minute readings where $t$ still denotes the day number.
Thus we aim to over-sample the original time series by $5$-fold.
The difference from the study in section \ref{sec:study1} is we need more than $5$ models to solve the problem since 30 minutes is not a valid cycle for energy consumption data.
Instead, we still use one day as the cycle, and thus we consider each 6-minute period in a day separately which requires a total of $24\times60/6=240$ separate AR models.
Due to the high correlation between $Y_t^{(i)}$ in index $i$, we considered a Farlie-Gumbel-Morgenstern (FGM) copula model to capture the correlation.
\begin{definition}[FGM Copula]
Let $\bm{U}:=(U_1,\dots,U_m)$ be a random vector following a $m$-variate FGM copula, where $U_i$ takes value in $[0,1]$ for each $i$. 
The FGM copula is parameterized by $\theta\in\mbb{R}^{2^m-m-1}$ where we index each dimension of $\theta$ by a non-empty, non-singleton subset of $\{1,\dots,m\}$.
The joint distribution function is given by
\begin{align*}
    C_m(u_1,\dots,u_m;\theta) = &\, \mbb{P}(U_1\leq u_1,\dots, U_m\leq u_m) \\
    = &\, \prod_{i=1}^m u_i \left(1 + \sum_{    k=2}^m \sum_{1\leq j_1<\cdots<j_k\leq m} \theta_{j_1,\dots,j_k}(1-u_{j_1})\cdots(1-u_{j_k})\right),
\end{align*}
\end{definition}
In practice, we generate $u_i$ from the copula distribution where $u_i$ represents the quantile position of the $i$-th component $Y^{(i)}$. 
Then $u_i$ is transformed into a sample for $Y_t^{(i)}$ through the inverse transformation of the marginal distribution function.
The marginal distributions of the copula model can still be captured by (\ref{eq:ARmodel}):
\begin{equation}\tag{\ref{eq:ARmodel}}
 Y_{t}^{(i)} \sim  \GenLog\left(\alpha^{(i)}, \beta^{(i)}, \gamma^{(i)}, C^{(i)} + \mu_t^{(i)}\right), \;\;\;\;\;\;
\mu_{t}^{(i)} = \sum_{r=1}^K \Phi_r^{(i)} Y_{t-r}^{(i)} + \bmXi_t \bm{\psi}^{(i)}
\end{equation}
where $i$ ranges from $1$ to $240$.
Each day has $48$ half-hours and each half-hour induces a constraint, $j=1,\dots,48$,
$$
\mcal{H}_t^{(j)}:= \left\{\left(Y_t^{(5(j-1)+1)},\dots, Y_t^{(5j)}\right): S_t^{(j)} = \sum_{i=1}^5 Y_t^{(5(j-1)+i)}\right\}.
$$
Weather data near the power station is used as additional covariates $\bmXi_t$ which include an hourly temperature and humidity reading.
Due to the copula structure, instead of having a product target density of $\prod_i f_i(y^{(i)})\indi_{\mcal{H}}$, the target density is a single density function with constraint $f^*(y^{(1)},\dots,y^{(m)})\indi_{\mcal{H}}$ where $f^*$ denotes the density function of the copula.
In other words, the collection $Y^{(1)},\dots,Y^{(m)}$ is considered as a single random variable of dimension $m$, yet, sampling from $f^*$ without constraint is still simple.

\subsubsection*{Result}
\begin{figure}[t]
    \centering
    \includegraphics[width=0.75\linewidth]{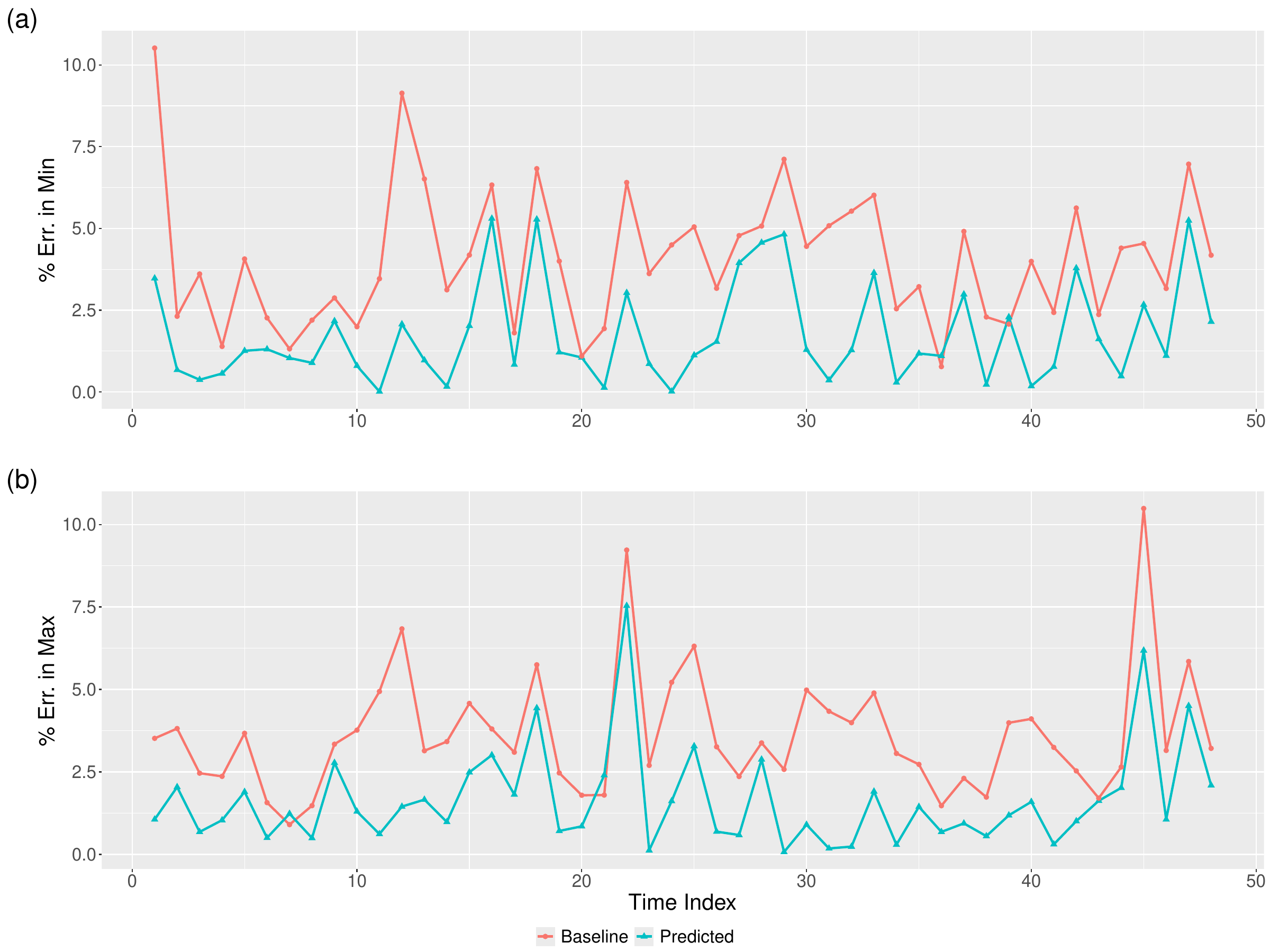}
    \caption{Percentage difference prediction of peak and trough values, comparing the constrained model with the baseline.}
    \label{fig:max_min}
\end{figure}
We used 3 months of high-frequency reading (from June to September 2021) to estimate the $240$ sets of parameters like in the first study.
Given the low-frequency readings for the subsequent 3 days, we imputed the high-frequency readings and extracted the peaks and troughs for each 30-minute period.
The peak and trough prediction performance is compared against the naive baseline which uses the 30-minute readings for both peak and trough estimates.
The result of peak and trough predictions is plotted in Fig. \ref{fig:max_min}, where red circles represent the baseline, green triangles represent the true values and blue squares represent the predicted values.
We can see clearly that the predicted values are closer to the true values than the baseline. 
To be more specific, the RMSE of our predicted values is only $55\%$ of the baseline RMSE.

\subsection{Study 3: Medicine Price Disaggregation}\label{sec:study3}
In this situation, we consider the problem of imputing the price of a certain generic medicine in different pharmacies given the sample mean and standard variation. Generic medicines are patent expired medicines which can be made by any company and most of the generic manufacturing now takes place in India and China. Once such imported generic medicine is approved by Medicines and Healthcare products Regulatory Agency (MHRA), any pharmacy can get their money back when the pharmacy gives the medicine to a patient, at a reimbursement rate. However, the pharmacy's purchasing price of the medicine may be much higher than the reimbursement rate due to short supply, which could lead to losses to the chemist. Therefore, it is important for the government to monitor the purchasing price regularly and introduce a different reimbursement price for medicines in short supply. In practice, we usually can obtain the purchasing price time series data from the past and due to privacy reasons, the recent months' data are given only in the form of summary statistics, i.e., mean and variance.
Let $Y^{(i)}_t$ be the price of the medicine in month $t$ at pharmacy $i$, then the constraints are given by
$$
 \mcal{H}_t:=\left\{\frac{1}{m}\sum_{i=1}^m Y^{(i)}_t = \mu_t,\qquad \frac{1}{m}\sum_{i=1}^m (Y^{(i)}_t-\mu_t)^2 = S_t\right\}.
$$
The goal is to fit a time series model for each $Y^{(i)}_t$ and sample for each month $t$ given the pair $(\mu_t,S_t)$.

\subsubsection*{Parameter Estimation}
For model fitting, we are using a similar autoregressive setting as in (\ref{eq:ARmodel}), except in this case we are modelling the error in price as Student's t-distributions and we use no extra covariates.
We assume the time series are independent with no additional covariance structure imposed.

\subsubsection*{Result}
For disaggregation we applied Algorithm \ref{alg:CMCF-2} corresponding to the second approach in Section \ref{sec:proposal_sampling} where instead of generating from constrained Gaussian, we sample the points uniformly from the constraint and apply a rejection step based on the Gaussian density.
We used the CHMC sampler proposed in \cite{lelievre2019hybrid} to generate points uniformly from the constraint.
\begin{figure}[t]
    \centering
    \includegraphics[width=0.85\linewidth]{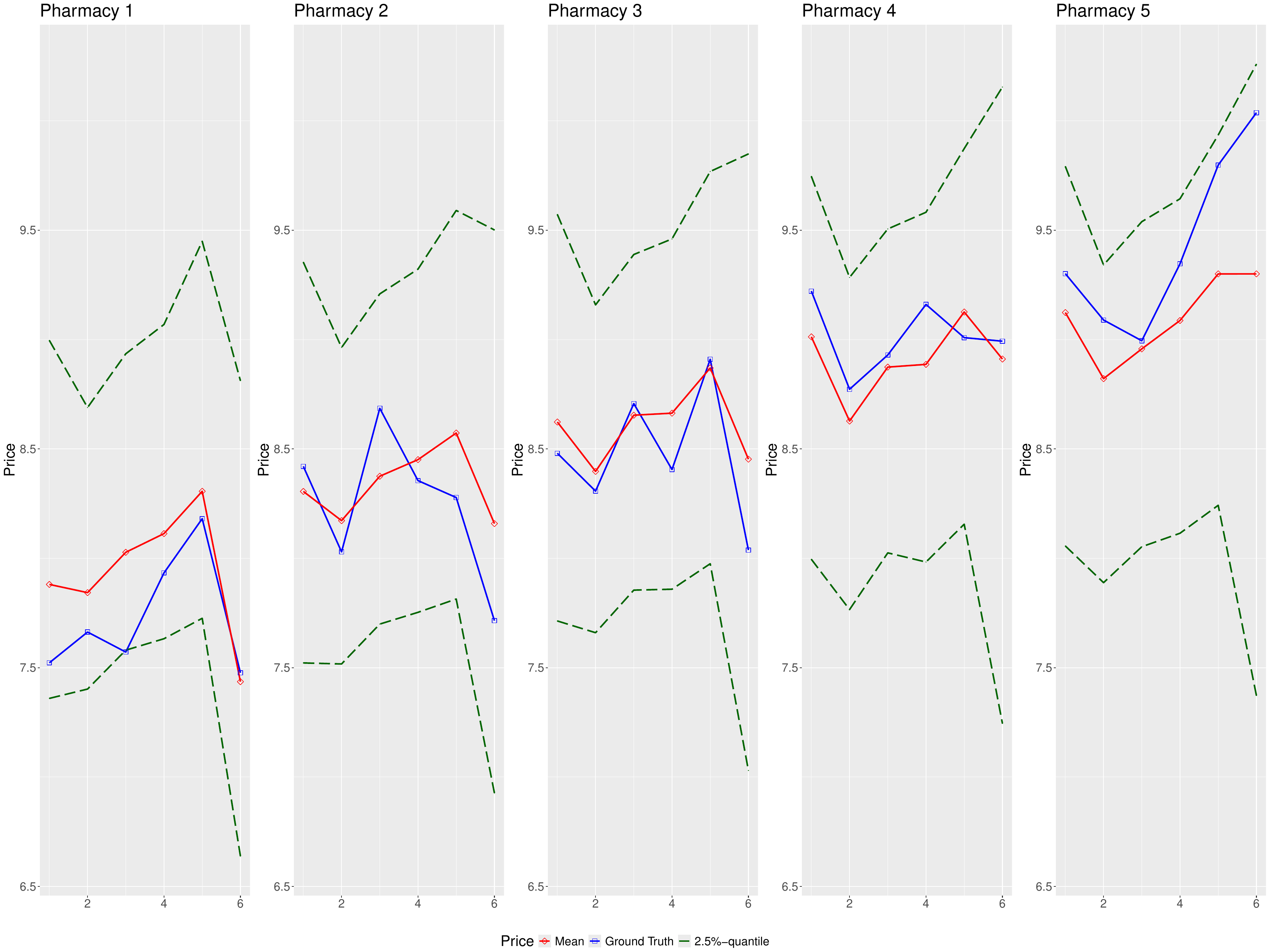}
    \caption{Medicine price disaggregated given sample mean and covariance.}
    \label{fig:medicine}
\end{figure}
Fig. \ref{fig:medicine} plots the disaggregated result given mean and sample variance. 
The mean value (in red) and $2.5\%,97.5\%$-quantiles (in green) are computed from the simulated samples and plotted against the ground truth (in blue).
Notably, since the sample variance is known in the simulation, the quantile lines quite precisely capture the uncertainties and at some point coincide with the ground truth value.

\section{Mean-Squared Error analysis}\label{appx:MSE}

\subsection{MSE analysis}
In this section, we consider a toy Gaussian example and use it to demonstrate why and when adding linear constraints can improve the accuracy of model forecasting (or using words `accuracy of data imputation').
For simplicity, we will focus on a toy Gaussian example without considering time-series data. 

Consider a simple linear regression setting.
Let $y^{(i)}\sim p(\cdot)$, $i=1,\cdots,m$, to be $n$ independent observations from the true population.
Through the regressoin model, each observation $y^{(i)}$ has a corresponding estimator $\hat{Y}^{(i)}$ with
$
 \hat{Y}^{(i)}\sim \mcal{N}\left(\hat{\mu}_i,\hat{\sigma}_i^2\right)
$
where $\hat{\mu}_i$ is the fitted mean of $y^{(i)}$.
Here we do not assume that $\hat{Y}^{(i)}$ is necessarily related to the target $y^{(i)}$ it intends to predict, i.e. $\hat{Y}^{(i)}$ may not be a sensible predictor of $y^{(i)}$ that may have a very large bias or uncertainty.
Consider another Normal random variable $S$ where
$$\hat{S} := \sum_{i=1}^m \hat{Y}^{(i)}, $$
then the joint distribution of $(\hat{Y}_1,\dots,\hat{Y}_m,\hat{S})$ is still a Multivariate Normal distribution with mean and variance given by
$$
 \hat{\bm{\mu}} = \begin{bmatrix}
        \hat{\mu}_1,&\cdots, &\hat{\mu}_m, & \sum_{i=1}^m\hat{\mu}_i
      \end{bmatrix}^{\top},$$ 
and 
$$ \hat{\bm{\Sigma}} = [
            \bm{D},   \bm{V} ;
            \bm{V}^\top, w^2 ]
$$
where 
$$\bm{D} := \text{diag}(\hat{\sigma}_1^2,\cdots, \hat{\sigma}_m^2),\quad \bm{V}:=\left[\hat{\sigma}_1^2,\cdots,\hat{\sigma}_m^2\right]^{\top},\quad w^2 := \sum_{i=1}^m\hat{\sigma}_i^2.$$
Let $\hat{\bm{Y}}$ denote the random vector of $(\hat{Y}_1,\dots,\hat{Y}_m)$, then its distribution conditioned on the constraint is
$
 \hat{\bm{Y}}\left|\hat{S}=s \right. \sim \mcal{N}\left(\mu^{*},\Sigma^{*}\right)
$
where
\begin{equation}
\bm{\mu}^*(s)  =  \left[\hat{\mu}_1 + \frac{\hat{\sigma}_1^2}{w^2}(s-\sum_{i=1}^m \hat{\mu}_i), \cdots, \hat{\mu}_m + \frac{\hat{\sigma}_n^2}{w^2}(s-\sum_{i=1}^m\mu_i)\right]^\top, \quad \bm{\Sigma}^*(s) = \bm{D} - \left[\frac{\hat{\sigma}_i^2\hat{\sigma}_j^2}{w^2}\right]_{i,j}^{1,m}.
\end{equation}
We examine the forecasting (imputation) performance in three aspects:
\begin{enumerate}
 \item \emph{Residue} of each predictor $\hat Y^{(i)}$ computed by $\alpha_i:=y^{(i)} -\hat{\mu}_i$
 \item \emph{Uncertainty} of the predictors computed by $\text{Var}(\hat{Y}^{(i)})$
 \item \emph{Mean-squared error} (MSE) computed by $\mbb{E}\left[\|\bm{\hat{Y}}-\bm{y}\|^2\right]$
\end{enumerate}
where MSE, measures both mean deviation and uncertainty, is a more comprehensive evaluation of the performance.
The potential improvement in MSE if the sum of the true values is incorporated into the model is given by:
\begin{equation}
  \mbb{E}\left[\|\hat{\bm{Y}}-\bm{y}\|^2\right] - \mbb{E}\left[\|\hat{\bm{Y}}-\bm{y}\|^2 \left| S = \sum_{i=1}^{m}\hat{Y}^{(i)}\right.\right] 
  =\underbrace{\sum_{i=1}^m \alpha_i^2-\sum_{i=1}^m\left(\alpha_i-\frac{\hat{\sigma}_i^2}{w^2}\sum_{j=1}^m\alpha_j\right)^2}_{\Psi_1} + \underbrace{\frac{1}{w^2}\sum_{i=1}^{m}(\hat{\sigma}_i^4)}_{\Psi_2}. \label{eq:mse_appen}
\end{equation}

\begin{remark}
 We can observe from $\Psi_1$ the formula for the constrained residue, namely $\alpha_i-\frac{\hat{\sigma}_i^2}{w^2}\sum_{j}\alpha_j$.
 Thus the correction on residue is in the favourable direction only for the components whose $\alpha_i$ has the same sign as $S_d:= \sum_{i=1}^m \alpha_i$ the overall deviation from the constraint.
 Thus, we expect to see the mean deviation improving in all dimensions if all components have over(under)-estimated the target.
 However, it is still possible for components with small $\alpha_i$ but large relative variance $\lambda_i:= \hat{\sigma}_i^2/w^2$ to be over-corrected and end up with a worse mean deviation.
\end{remark}
\begin{remark} 
 When the bias term $\Psi_1$ is not dominated by the variance term $\Psi_2$, then $\Psi_1$ can be negative and consequently  rendering (\ref{eq:mse_appen}) negative. 
 $\Psi_1$ can be viewed as a quadratic function of $\alpha_i$ for every $i$ and when the leading coefficient is negative, the function is more likely to take a value below zero. One sufficient condition for this to happen is when $\lambda_i:=\hat{\sigma}_i^2/w^2 <1-\sqrt{(m-1)/m}$.
\end{remark}
\begin{remark} \label{remark:special_case}
  One special case is when $\lambda_1=\cdots=\lambda_m$, in which case $\Psi_1$ is guaranteed to be non-negative and the constrained model is always better than the unconstrained model in terms of MSE. 
\end{remark}
We can deduce the following Propositions from equation (\ref{eq:mse_appen}).

\setcounter{theorem}{4}
\begin{proposition}[Garaunteed Uncertainty Reduction] \label{prop:domination}
The sum of constrianed variance $\text{tr}(\Sigma^*)$ is always less than the unconstrained variance with a reduction of
\[
\text{tr}(\bm{\Sigma}^*)-\sum_{i=1}^m \hat{\sigma}_i^2 = \frac{1}{w^2}\sum_{i=1}^{m}\hat{\sigma}_i^4, \quad \text{ where } w^2=\sum_{i=1}^m \hat{\sigma}^2_i.
\]
\end{proposition}
\begin{proof}
Follows directly from (\ref{eq:mse_appen}) since $\Psi_2$ denotes exactly the difference in uncertainty and $\Psi_2\geq0$.
\hfill$\square$
\end{proof}
It is important to note that, we made no assumption about the imputed value $\hat {\bm Y} = (\hat{Y}^{(1)}, \cdots, \hat{Y}^{(m)})$, except that it is independent in its components and is Gaussian.
If we could also bound its error, i.e. having a reasonable statistical model, then we can deduce a stronger result as follows:

\begin{proposition}[Uncertainty Domination]\label{prop:uncertainty}
Let $\alpha_i:=y^{(i)}-\hat{\mu}_i$  and $S_d$ denote the sum of $\alpha_i$.
Suppose that the predictors are reasonable that 
\[\exists M>0 \text{ such that } \forall i, |\alpha_i|\leq \lambda_i M, w^2\geq 2S_dM,\]
then the constrained model always has a lower mean-squared error compared with the unconstrained model.
\end{proposition}
\begin{proof}
We try to bound $\Psi_1$ in (\ref{eq:mse_appen}).
Firstly recall that 
\[
S_d = \sum_{i=1}^m \alpha_i
\]
Without loss of generality, let $S_d\geq 0$, then 
\begin{align*}
    \Psi_1 & =  -2S_d\sum_{i=1}^m \lambda_i\alpha_i + \sum_{i=1}^m \lambda_i^2 S_d^2\\
    & \geq  - 2S_d M\sum_{i=1}^m \lambda_i^2
\end{align*}
The second line follows by applying the bound on $\alpha_i$.
Note that
\[
\Psi_2 = w^2\sum_{i=1}^m \lambda_i^2.
\]
Thus
\begin{align*}
\Psi_1+\Psi_2 & \geq  w^2\sum_{i=1}^m\lambda_i^2 - 2S_d M\sum_{i=1}^m \lambda_i^2 \\
& \geq  0
\end{align*}
\hfill$\square$
\end{proof}
The main takeaway from Proposition \ref{prop:domination} is very similar to Proposition \ref{prop:uncertainty}, namely, it is most appropriate to apply a constrained model when the original model has a higher uncertainty compared to its expected error.

\subsection{Effect of relative variance on MSE}
Since there is a guaranteed improvement in the variance component $\Psi_2$, it is interesting to analyse what happens if $\Psi_1$ is not dominated by $\Psi_2$.
Define the following
\[
\lambda_i=\hat{\sigma_i}^2/w^2, \qquad \Lambda:=\sum_{i=1}^m\lambda_i^2
\]
\begin{align}
\Psi_1(\bm{\alpha}) &= \sum_{i=1}^m \alpha_i^2 - \sum_{i=1}^m\left(\alpha_i-           \frac{\sigma_i^2}{w^2}\sum_{j=1}^m\alpha_j\right)^2 \nonumber\\
    &= 2(\sum_{i=1}^m\alpha_i)\sum_{j=1}^m\alpha_j\lambda_j - \left(\sum_{i=1}^m\alpha_i\right)^2\sum_{j=1}^m\lambda_j^2 \nonumber\\
    &=  2\sum_{i=1}^m\alpha_i^2\lambda_i + 2\sum_{1\leq i<j}^{m}\alpha_i\alpha_j (\lambda_i+\lambda_j) - \sum_{i=1}^m\lambda_j^2\left(\sum_{i=1}^m\alpha_i^2 + \sum_{1\leq i<j}^{m}2\alpha_i\alpha_j\right)\nonumber\\
    &=\sum_{i=1}^m(2\lambda_i-\Lambda)\alpha_i^2 + 2\sum_{1\leq i<j}^m(\lambda_i+\lambda_j-\Lambda)\alpha_i\alpha_j \nonumber
\end{align}
Take $\alpha_1$ for example, the roots lie at
\begin{equation*}
    \alpha_1 = \frac{1}{2(\lambda_1-\Lambda)}\left(-2\sum_{i=2}^{m}(\lambda_1+\lambda_i-\Lambda)\alpha_i \pm \sqrt{\Delta}\right)
\end{equation*}
where
\begin{equation*}
    \Delta = \left(\sum_{i=2}^m (\lambda_1-\lambda_i)\alpha_i\right)^2 \geq 0.
\end{equation*}
Thus, the equation has a repeated root if and only if $\lambda_1=\cdots=\lambda_n$. 
In other words, depending on the mean deviation $\alpha_i$, it is almost always possible for $\Psi_1$ to be negative and have negative improvement for the overall MSE.

\begin{figure}[t!]
    \centering
    \makebox[\textwidth]{
    \begin{tabular}{  @{}r@{}lcr@{} c@{} lcr c@{} lcr c@{} lcr @{}}\toprule
     \multicolumn{4}{c}{Gaussian} & \phantom{====} & \multicolumn{3}{c}{MSE Improv.} &\phantom{a} & \multicolumn{3}{c}{Devi. Improv.} &\phantom{a} & \multicolumn{3}{c}{Var. Improv.}\\\cmidrule{1-4}\cmidrule{6-16}
     \splitcell{c}{$\bm{\alpha}=$\phantom{=} \\ $\bm{\sigma}^2=$\phantom{=}} & \splitcell{l}{(0.1 \\ (1} & \splitcell{c}{-1 \\ 4} & \splitcell{r}{2) \\ 10)} & & [\underline{0.08} & \underline{0.39} & \underline{9.06}] & &  [\underline{0.07} & \textbf{-0.29} & \underline{0.73}]  & & [\underline{0.33} & \underline{1.34} & \underline{3.34}]  \\\midrule
     \splitcell{c}{$\bm{\alpha}=$\phantom{=} \\ $\bm{\sigma}^2=$\phantom{=}} & \splitcell{l}{(0.1 \\ (1} & \splitcell{c}{-2 \\ 4} & \splitcell{r}{4) \\ 10)} & & [\underline{0.08} & \textbf{-1.49} & \underline{15.9}] & &  [\underline{0.06} & \textbf{-0.56} & \underline{1.39}]  & & [\underline{0.33} & \underline{1.34} & \underline{3.34}] \\\midrule
     \splitcell{c}{$\bm{\alpha}=$\phantom{=} \\ $\bm{\sigma}^2=$\phantom{=}} & \splitcell{l}{(1 \\ (1} & \splitcell{c}{-3 \\ 4} & \splitcell{r}{8) \\ 10)} & & [\underline{0.70} & \textbf{-11.1} & \underline{54.7}] & & [\underline{0.40} & \textbf{-1.60} & \underline{4.00}]  & & [\underline{0.33} & \underline{1.34} & \underline{3.34}] \\\midrule
     \splitcell{c}{$\bm{\alpha}=$\phantom{=} \\ $\bm{\sigma}^2=$\phantom{=}} & \splitcell{l}{(10 \\ (1} & \splitcell{c}{-3 \\ 2} & \splitcell{r}{-5) \\ 3)} & & [\underline{6.78} & \textbf{-3.80} & \textbf{-9.51}] & & [\underline{0.34} & \textbf{-0.67} & \textbf{-1.00}] & & [\underline{0.33} & \underline{0.67} & \underline{1.00}]\\\midrule
     \splitcell{c}{$\bm{\alpha}=$\phantom{=} \\ $\bm{\sigma}^2=$\phantom{=}} & \splitcell{l}{(10 \\ (2} & \splitcell{c}{-3 \\ 2} & \splitcell{r}{-5) \\ 2)} & & [\underline{13.6} & \textbf{-3.79} & \textbf{-6.46}] & & [\underline{0.67} & \textbf{-0.67} & \textbf{-0.67}] & &  [\underline{0.67} & \underline{0.67} & \underline{0.67}]\\\midrule
     \multicolumn{4}{c}{Student's T ($\nu$=5)} & \multicolumn{12}{c}{}\\\cmidrule{1-4}
     \splitcell{c}{$\bm{\alpha}=$\phantom{=} \\ $\bm{\sigma}^2=$\phantom{=}} & \splitcell{l}{(0.1 \\ (1} & \splitcell{c}{-2 \\ 4} & \splitcell{r}{4) \\ 10)} & & [\underline{0.28} & \underline{0.26} & \underline{21.2}] & &  [\underline{0.02} & \textbf{-0.60} & \underline{1.33}] & & [\underline{0.28} & \underline{3.02}& \underline{12.2}]\\\midrule
     \splitcell{c}{$\bm{\alpha}=$\phantom{=} \\ $\bm{\sigma}^2=$\phantom{=}} & \splitcell{l}{(1 \\ (1} & \splitcell{c}{-3 \\ 4} & \splitcell{r}{8) \\ 10)} & & [\underline{0.70} & \textbf{-11.8} & \underline{57.3}] & &  [\underline{0.46} & \textbf{-1.73} & \underline{3.82}]  & & [\underline{0.01} & \underline{1.48} & \underline{10.6}] \\\midrule
     \splitcell{c}{$\bm{\alpha}=$\phantom{=} \\ $\bm{\sigma}^2=$\phantom{=}} & \splitcell{l}{(10 \\ (1} & \splitcell{c}{-3 \\ 2} & \splitcell{r}{-5) \\ 3)} & & [\underline{7.85} & \textbf{-3.09} & \textbf{-7.32}] & & [\underline{0.38} & \textbf{-0.69} & \textbf{-0.94}]  & & [\underline{0.45} & \underline{1.53} & \underline{2.98}]\\\midrule
     \splitcell{c}{$\bm{\alpha}=$\phantom{=} \\ $\bm{\sigma}^2=$\phantom{=}} & \splitcell{l}{(10 \\ (2} & \splitcell{c}{-3 \\ 2} & \splitcell{r}{-5) \\ 2)} & & [\underline{14.3} & \textbf{-2.90} & \textbf{-5.55}] & & [\underline{0.66} & \textbf{-0.66} & \textbf{-0.67}] & &  [\underline{1.53} & \underline{1.53} & \underline{1.54}]\\\midrule
     \multicolumn{4}{c}{Gen. Logistic} & \multicolumn{12}{c}{}\\\cmidrule{1-4}
     \multicolumn{4}{l}{Set 1, $\bm{\alpha}=(10,\,\,-3,\,\, -5)$} & & [\textbf{-15.3} & \underline{3.45} & \underline{4.15}] & & [\underline{0.35} & \textbf{-0.82} & \textbf{-0.82}] & & [\underline{0.52} & \underline{2.07} & \underline{2.07}] \\\midrule
     \multicolumn{4}{l}{Set 2, $\bm{\alpha}=(10,\,\,-3,\,\, -5)$} & & [\underline{16.4} & \underline{12.8} & \underline{5.73}] & & [\underline{1.49} & \textbf{-0.25} & \textbf{-0.25}] & & [\underline{24.3} & \underline{3.26} & \underline{3.26}] \\
     \bottomrule
    \end{tabular}}
    \caption{Improvements in accuracy when adding sum constraint for different cases. Improvement in MSE, deviation and variance for all three components of the model are listed with positive values marked by an underscore and negative values marked in bold.}
    \label{tab:MSE1}
\end{figure}

\subsection{Gaussian Case}
A simulation is conducted to demonstrate the analysis made above, see Fig. \ref{tab:MSE1}.
A total of 100,000 samples are generated from both the constrained and unconstrained models in each setting. In the first case, the mean deviation is dominated by the variance and we see improvements in all three predictors (or called as imputed values).  
For the second and third settings, the variances are kept the same but the mean deviations are increased.
We can see the variance improvement are the same but MSE improvements are not all positive.
However, the total improvement in model MSE is still positive, mainly because the correction for the third component which has the largest mean error has a large positive effect on the model MSE.
In these two cases, the main contribution to improvement in overall MSE is no longer the reduction in variance (uncertainty).

Finally, for cases 4 and 5, we examine the situation when $\alpha_i$ are large but $\sigma_i^2$ are small.
Note that for the fourth case, the first predictor has a very small relative variance and is below the $1-\sqrt{(m-1/m)}$ threshold, thus the overall MSE improvement is negative.
This is different in the fifth case, when the variance in all components is the same, which matches the condition in Remark \ref{remark:special_case} and we observe a small but positive improvement in MSE.

\subsection{T-Distribution Case}
When the modeling distribution is non-Gaussian, the constrained distribution becomes intractable and hence there is no analytic formula for measuring the MSE improvement.
We have done the same simulation on Student's t-distribution.
$\bm{\alpha}$ and $\bm{\sigma^2}$ still represent the mean deviation and variance respectively in the T-distribution's case.

The generalized Student's T-distributions are implemented with degrees of freedom fixed to 5 and the density function can be expressed as
\[    
f(x;\nu, \mu,\sigma^2)=\frac{\Gamma(\frac{\nu+1}{2})}{\sqrt{\nu\pi}\Gamma(\frac{\nu}{2})}\left(1+ \frac{(x-\mu)^2}{\nu\sigma^2}\right)^{-\frac{\nu+1}{2}}
\]
The four cases examined for Student's t-distribution use the same parameter setting as cases 2-5 in the Gaussian simulation and the results almost match case by case.
We see similar improvements in deviation estimation but overall larger improvements in uncertainty when applying a constraint to Student's t-distribution compared with the Gaussian cases.
This is reasonable since Student's t-distribution has a heavier tail than the Gaussian distribution, and hence a larger uncertainty when the scaling parameter $\sigma^2$ is the same.

\subsection{Generalized Logistic Distribution Case}
The Generalized Logistic distribution has too many parameters, so the setting is omitted from the table.
The first setup assumes that $X_1\sim\GenLog(3,1.2,1,0)$, $X_2,X_3\sim\GenLog(3,2,2,0)$. 
The second setup assumes that 
$X_1\sim\GenLog(3,0.4,2,-5)$, $X_2\sim\GenLog(3,0.4,1,-2)$ and $X_3\sim\GenLog(3,0.4,1,-3)$.
To clarify, the $\bm{\alpha}$ stated in Table \ref{tab:MSE1} still represents the mean deviation of the estimation model from the true value.
In both setups, the true values are chosen such that the estimation model is off by exactly $10$, $-3$, and $5$ in the corresponding dimension.

The result is similar to the Gaussian case even though the distributions are skewed. 
The random variables in setup 1 have much lower variance than in setup 2.
Note that setup 2 is highly positively skewed with an extremely heavy tail towards the positive end.
Similar to setups 4 and 5 in Gaussian and T-distribution cases, we see a positive overall improvement in MSE when the model uncertainty is large and the deviation is not guaranteed to improve.
When the model uncertainty is small, the overall MSE actually becomes larger due to the inaccurate model.

\begin{remark}
    Overall, the results derived from the Gaussian case carry over relatively well to the non-Gaussian cases. One may expect better MSE performance as long as the unconstrained model captures the true value in its typical region with a good estimate of model uncertainty.
\end{remark}

\end{appendix}







\section*{Funding Information}
SH, HD, LA, MP, GOR were supported by the EPSRC research grant "Pooling INference and COmbining Distributions Exactly: A Bayesian approach (PINCODE)", reference (EP/X028100/1, EP/X028119/1, EP/X028712/1, EP/X027872/1).

LA, HD, MP and GOR were also supported by UKRI for grant EP/Y014650/1, as part of the ERC Synergy project OCEAN.

GOR was also supported by EPSRC grants Bayes for Health (R018561), CoSInES (R034710), and EP/V009478/1.

\bibliographystyle{abbrvnat}
\bibliography{bibtex}

\nocite{adhikari2013introductory}
\nocite{nicolaescu2020lectures}

\end{document}